\def\bX{\mathbf{X}}
\def\bx{\mathbf{x}}
\def\bXt{\mathbf{X}^\mathrm{T}}

\def\bXz{\mathbf{X}_0}
\def\bXtz{\mathbf{X}_0^\mathrm{T}}
\def\bxz{\mathbf{x}_0}
\def\bxtz{\mathbf{x}_0^\mathrm{T}}
\def\bP{\mathbf{P}}
\def\bPz{\mathbf{P}_0}
\def\bp{\mathbf{p}}
\def\bpz{\mathbf{p}_0}
\def\bptz{\mathbf{p}_0^\mathrm{T}}
\def\bPt{\mathbf{P}^\mathrm{T}}
\def\bPtz{\mathbf{P}_0^\mathrm{T}}
\def\bpt{\mathbf{p}^\mathrm{T}}

\def\ra{\mathrm{a}}
\def\rb{\mathrm{b}}
\def\rc{\mathrm{c}}
\def\rd{\mathrm{d}}
\def\re{\mathrm{e}}
\def\rf{\mathrm{f}}
\def\rn{\mathrm{n}}
\def\rr{\mathrm{r}}
\def\rlm{\mathrm{m}}
\def\rA{\mathrm{A}}
\def\rB{\mathrm{B}}
\def\rC{\mathrm{C}}

\def\rE{\mathrm{E}}

\def\rR{\mathrm{R}}
\def\rS{\mathrm{S}}

\def\rC{\mathrm{C}}
\def\rE{\mathrm{E}}
\def\rS{\mathrm{S}}
\def\rT{\mathrm{T}}
\def\tr{\mathrm{Tr}}

\def\bA{\mathbf{A}}

\def\bB{\mathbf{B}}

\def\bC{\mathbf{C}}

\def\bH{\mathbf{H}}

 \def\basr{\mathbf{a}(\nur)}

\def\blf{\mathbf{f}}
\def\bfdot{\dot{\mathbf{f}}}

\def\rsum{\sum_{\rr=1}^{\infty}}

\def\bnu{\boldsymbol{\nu}}
\def\bxi{\boldsymbol{\xi}}

\def\bG{\mathbf{G}}
\def\bGt{\mathbf{G}^\mathrm{T}}
\def\bGdot{\dot{\mathbf{G}}}

\def\bGhat{\hat{\mathbf{G}}}
\def\bGhatnur{\hat{\mathbf{G}}(\nur)}

\def\bfdot{\dot{\mathbf{f}}}

\def\bfhat{\hat{\mathbf{f}}}
\def\bfhatnur{\hat{\mathbf{f}}(\nur)}
\def\blm{\mathbf{m}}
\def\bmtomo{\mathbf{m}^{-1}}

\def\bmtomoot{\mathbf{m}^{-\frac{1}{2}}}
\def\bM{\mathbf{M}}
\def\bg{\mathbf{g}}
\def\bgt{\mathbf{g}^\mathrm{T}}
\def\bC{\mathbf{C}}
\def\bc{\mathbf{c}}
\def\bctomo{\mathbf{c}^{-1}}
\def\bgamma{\boldsymbol{\gamma}}
\def\bgammadot{\dot{\boldsymbol{\gamma}}}

\def\bgammahatnur{\hat{\boldsymbol{\gamma}}(\nur)}
\def\bomega{\boldsymbol{\omega}}

\def\nur{\nu_\mathrm{r}}
\def\del{\partial}
\def\lbr{\left\langle}
\def\rbr{\right \rangle_\rC}

\def\sinomega{\sin(\boldsymbol{\omega}}
\def\cosomega{\cos(\boldsymbol{\omega}}

\documentclass[aps,reprint,pdftex,nofootinbib,superscriptaddress]{revtex4-1}
\usepackage{amsmath,amssymb,amsthm,undertilde}
\usepackage{graphicx,color}
\usepackage[pdftex,colorlinks]{hyperref}

\newcommand{\bra}[1]{\left\langle #1 \right|}
\newcommand{\ket}[1]{\left| #1 \right\rangle}

\newtheorem{theorem}{Theorem}
\numberwithin{equation}{section}
\begin{document}

\title{The equilibrium states of open quantum systems in the strong coupling regime}
\author{Y. Suba\c{s}\i}
\email{ysubasi@umd.edu}
\affiliation{Joint Quantum Institute and Maryland Center for Fundamental Physics, University of Maryland, College Park, Maryland 20742}
\author{C. H. Fleming}
\email{hfleming@physics.umd.edu}
\affiliation{Joint Quantum Institute and Maryland Center for Fundamental Physics, University of Maryland, College Park, Maryland 20742}
\author{J. M. Taylor}
\email{jacob.taylor@nist.gov}
\affiliation{Joint Quantum Institute and NIST Gathersburg, Maryland}
\author{B. L. Hu}
\email{blhu@physics.umd.edu}
\affiliation{Joint Quantum Institute and Maryland Center for Fundamental Physics, University of Maryland, College Park, Maryland 20742}
\date{\today}

\begin{abstract}
In this work we investigate the late-time stationary states of open quantum systems coupled to a thermal reservoir in the strong coupling regime. In general such systems do not necessarily relax to a Boltzmann distribution if the coupling to the thermal reservoir is non-vanishing or equivalently if the relaxation timescales are finite. Using a variety of non-equilibrium formalisms valid for non-Markovian processes, we show that starting from a product state of the closed system = system + environment, with the environment in its thermal state, 
the open system which results from coarse graining the environment will evolve towards an equilibrium state at late-times. This state can be expressed as the reduced state of the closed system thermal state at the temperature of the environment. 
For a linear (harmonic) system and environment, which is exactly solvable, we are able to show in a rigorous way that all multi-time correlations of the open system evolve towards those of the closed system thermal state. Multi-time correlations are especially relevant in the non-Markovian regime, since they cannot be generated by the dynamics of the single-time correlations. For more general systems, which cannot be exactly solved, we are able to provide a general proof that all single-time correlations of the open system evolve to those of the closed system thermal state, to first order in the relaxation rates.
For the special case of a zero-temperature reservoir,
we are able to explicitly construct the reduced closed system thermal state in terms of the environmental correlations.
\end{abstract}

\maketitle

\section{Introduction}
Equilibrium states are typically discussed and derived in one of three settings or scenarios.
In the more-common equilibrium (Gibbs) perspective, originally based upon classical
ensemble theory, the entire system consisting of a system of interest plus its environment is taken to have some well-defined
state or set of states,
and upon coarse graining the environment, the system can appear thermal
\cite{Popescu06,Goldstein06}.
In the less-common non-equilibrium perspective, the environment is taken to be
initially thermal, whereas the open system is allowed to dynamically relax from
an arbitrary initial state into an equilibrium state
\cite{vanHove55,Davies74,Davies76a,Davies76}.
This approach is referred to as the Langevin paradigm \cite{CalzettaHu08}.
Both scenarios described above apply to situations where there is a clear 
distinction and separation between the system and environment degrees of freedom.
When there is no clear distinction or the separation is not physically justifiable, as in a molecular gas where each particle is identical, a very different set of physical variables and different kind of coarse-graining measure need be considered. One can examine the behavior of the n-particle distribution functions and perform the coarse-graining (e.g., 'slaving' in \cite{CalzettaHu08}) on  the Bogoliubov-Born-Green-Kirkwood-Yvon (BBGKY) hierarchy \cite{Balescu97}.
This approach is referred to as the Boltzmann paradigm.

The equilibrium and non-equilibrium perspectives can be made to complement each other rather
naturally within the Langevin or open system paradigm. 
In the former case, Popescu et al.~\cite{Popescu06} have shown that for an
overwhelming majority of pure states of the system + environment (within a
narrow energy interval), the reduced density matrix is very close to the
reduced density matrix corresponding to the microcanonical state of the system
+ environment (defined in the same energy interval). In their approach the
comparison is done without explicitly determining an equilibrium state. The
authors emphasize that for strong coupling, the equilibrium state is not of
Boltzmann type, and yet their results are valid in this domain. It is important
to note that dynamics does not play any role in their derivation; the entire
argument is based on \textit{kinematics}. The beauty of this approach is that one can explain the abundance of thermal-like states without referring to ensembles or time averages.

Linden \emph{et al} \cite{Linden09} expands upon the approach of \cite{Popescu06,Goldstein06} to demonstrate \textit{dynamical} relaxation\footnote{See Sec.~\ref{sec:terminology} for the definition of the terms relaxation, equilibration and thermalization as used in this work. There we also describe the meaning of the term \emph{equilibration} as used in Refs.~\cite{Linden09,Reimann10,Short11,Short12}, which differ from our definition.} under very weak assumptions.
Specifically,  they proved that any subsystem of a much larger quantum system will evolve to an approximately steady state. On the other hand Reimann \cite{Reimann10} showed that the expectation value of any ``realistic'' quantum observable will relax to an approximately constant value. (\cite{Short11} gave a clear analysis and unification of these two results.) Finally \cite{Short12} proves relaxation over a finite amount of time both in the sense of \cite{Linden09} and \cite{Reimann10}.

Dynamical relaxation of an open quantum system has been studied in the limit of vanishing coupling to the environment in \cite{vanHove55,Davies74,Davies76a,Davies76}.
In this limit the equilibrium state is shown to be of Boltzmann form in which case the result is called thermalization, rather than just relaxation. 
In our work reported here, we derive the equilibrium state of an open system coupled to a thermal reservoir explicitly, even in the strong coupling regime. Moreover for the N oscillator quantum Brownian motion (N-QBM) model we are able to show the relaxation of multi-time correlations of the open system as well. To do so we need to restrict the environment to be in a thermal initial state. 

Another difference between our work and \cite{Linden09,Reimann10,Short11,Short12} is in the methods and emphasis. We take the open quantum systems approach \cite{Feynman63,CaldeiraLeggett83,HPZ92,Davies76,Weiss93,Breuer02} of assuming an environment (E) which the system (S) interacts with, keeping some coarse-grained information about the environment and accounting for its systematic  influences on the system in a self-consistent manner. The time evolution of the open quantum system is in general governed by non-unitary dynamics. In contradistinction,  \cite{Linden09,Reimann10,Short12}  consider the unitary time evolution of the closed system (S + E) and then trace out the environment to get the system state. Both approaches are equally valid, each providing a different perspective into the physics with different emphasis. We will provide a more detailed comparison of our results  to those in the literature in the discussion section at the end.

\subsection{Relaxation, Equilibration and Thermalization}
\label{sec:terminology}

Before we present our approach, we want to define carefully what is meant by equilibration in this paper. To begin with let us consider a system in contact with two thermal reservoirs\footnote{We call an environment a reservoir if the environment has an infinite number of degrees of freedom, and a reservoir at constant temperature, a thermal reservoir.} at different temperatures. The system relaxes into a late-time steady state, which can be described by a reduced density matrix. All expectation values of system operators will also be time-independent at late-times. Yet there will be a steady heat flux from the hot reservoir to the cold reservoir through the system. This is an example of a non-equilibrium steady state. 

In general we define steady states via time independent density matrices: $d \boldsymbol{\rho}(t)/dt=0$ and use the term \emph{relaxation} to describe the generic convergence of the reduced density matrix to a fixed but arbitrary state in the late-time limit.
If the density matrix is diagonal in the energy eigenbasis of the system we call it a stationary state. An isolated stationary state is also a steady state, but this is not true for open systems with non-vanishing coupling to their environments.

In this work we reserve the term equilibrium for systems whose multi-time correlations can be derived from the thermal state of a possibly extended closed system which is governed by Hamiltonian dynamics.
As a result of our definition, equilibration implies relaxation but the reverse is not true.
The thermal reservoir distinguishes itself
from other possible environments by the \emph{universality} of its
fluctuation-dissipation relation (FDR)\footnote{As long as the environment is modelled after a physical system, fluctuations will be related to dissipation; hence there will be a FDR. However for general environments this relation depends on the specifics of the system-environment coupling. Thermal environments are unique in that the FDR does not depend on the details of the system and the coupling to the system \cite{FDR}.}, detailed-balance conditions and  Kubo-Martin-Schwinger (KMS) relations. In the vanishing coupling limit thermal reservoirs lead to the thermalization of the system as defined below.
However for non-vanishing coupling to a thermal reservoir the equilibrium state of the system does not need to be of the Boltzmann form
\begin{align}
\label{eq:TS}
\boldsymbol{\rho}_{\mathrm{S}}(\beta)=\frac{e^{-\beta \mathbf{H}_\mathrm{S}}}{\mathrm{Tr_S}\left[ e^{-\beta \mathbf{H}_\mathrm{S}} \right]}\, .
\end{align}
The asymptotic states we derive in this paper in the strong coupling limit describe equilibration and not thermalization.

The term thermalization is reserved for the relaxation of the density matrix of a system to the Boltzmann form \eqref{eq:TS} \textit{irrespective} of the initial state of the system.
Thermalization defined in this sense can take place only if the system-environment coupling is vanishingly weak. To be specific,
one requires (1) decaying environmental correlation functions, defined in
Sec.~\ref{sec:nonlinear}, (2) an initially thermal reservoir and (3) vanishing relaxation rates\footnote{To see a simple example of a relaxation rate consider the N-QBM model of Sec.~\ref{sec:Langevin} for N=1. In the Markovian limit the damping kernel can be written as $\boldsymbol{\gamma}(t,s)= \gamma_0 M \delta(t-s)$, where $\gamma_0$ acts as the damping rate.} or, equivalently, vanishing environmental correlation functions.

\begin{figure}[h]
\includegraphics[width=0.5\textwidth]{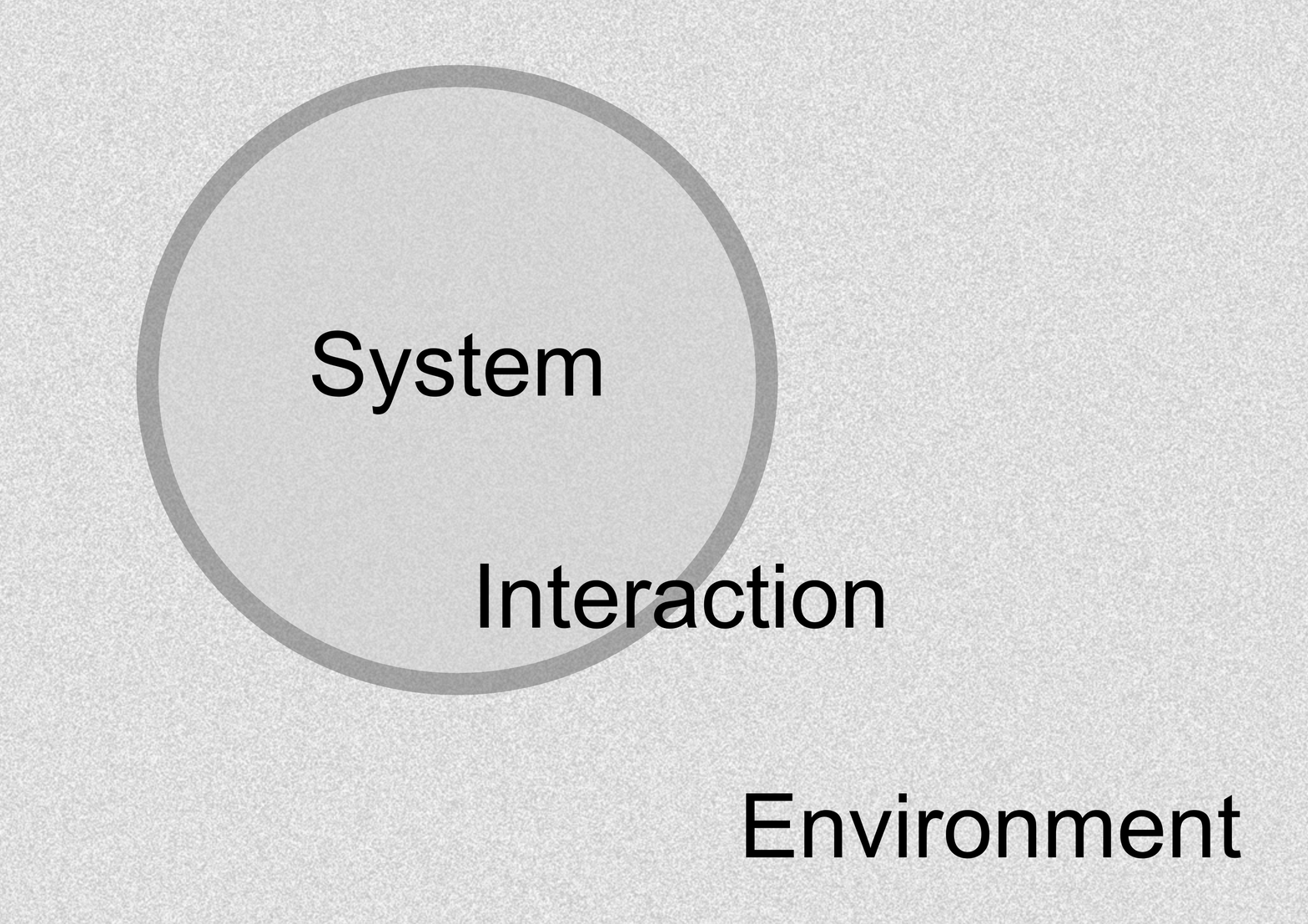}
\caption{\label{fig:boundary}
Depiction of a system embedded in its environment, with short-range interactions.
The typical argument for neglecting the interaction energy is that in the
macroscopic limit the boundary becomes immeasurable in relation to the bulk.}
\end{figure}

These conditions are customarily achieved by assuming short-range interactions and a relatively large system size, see Fig.~(\ref{fig:boundary}). However this assumption is generally not justifiable for small systems as Fig.~(\ref{fig:scale}) suggests.
In this paper, we address the stationary state of open quantum systems in contact with a thermal reservoir at temperature $T=1/\beta$, without the assumption of a vanishing interaction strength and allow for finite relaxation timescales ().
Relation \eqref{eq:TS} is known not to hold under these conditions \cite{Geva00}. Phenomenologically, one can estimate the corrections we describe by the ratio of the relaxation rates $\gamma$ to the system's energy level splittings $\Omega$, or $\gamma/\Omega$.\footnote{A well-known example is the density of states for an atom or molecule, which is
necessarily interacting with the electromagnetic field to a degree which cannot
be ignored when considering the Lamb shift, black-body radiation shifts, etc.. For optical frequencies, the emission rates of atoms are very small relative to their transition frequencies,
and so these corrections are very small.
However in other systems, such as condensates, these corrections can be of considerable size.}

As thoroughly discussed in Ref.~\cite{Geva00}, this fact is often overlooked in
many circumstances, due to the effects of ancillary approximations such as the
rotating-wave approximation, renormalization of environmentally-induced
energy-level shifts and overly-simplistic models.
As we explain in Appendix ~\ref{sec:CGN}, this fact may also be overlooked due to its absence in the case of classical, Gaussian noise.

\begin{figure}[h]
\includegraphics[width=0.5\textwidth]{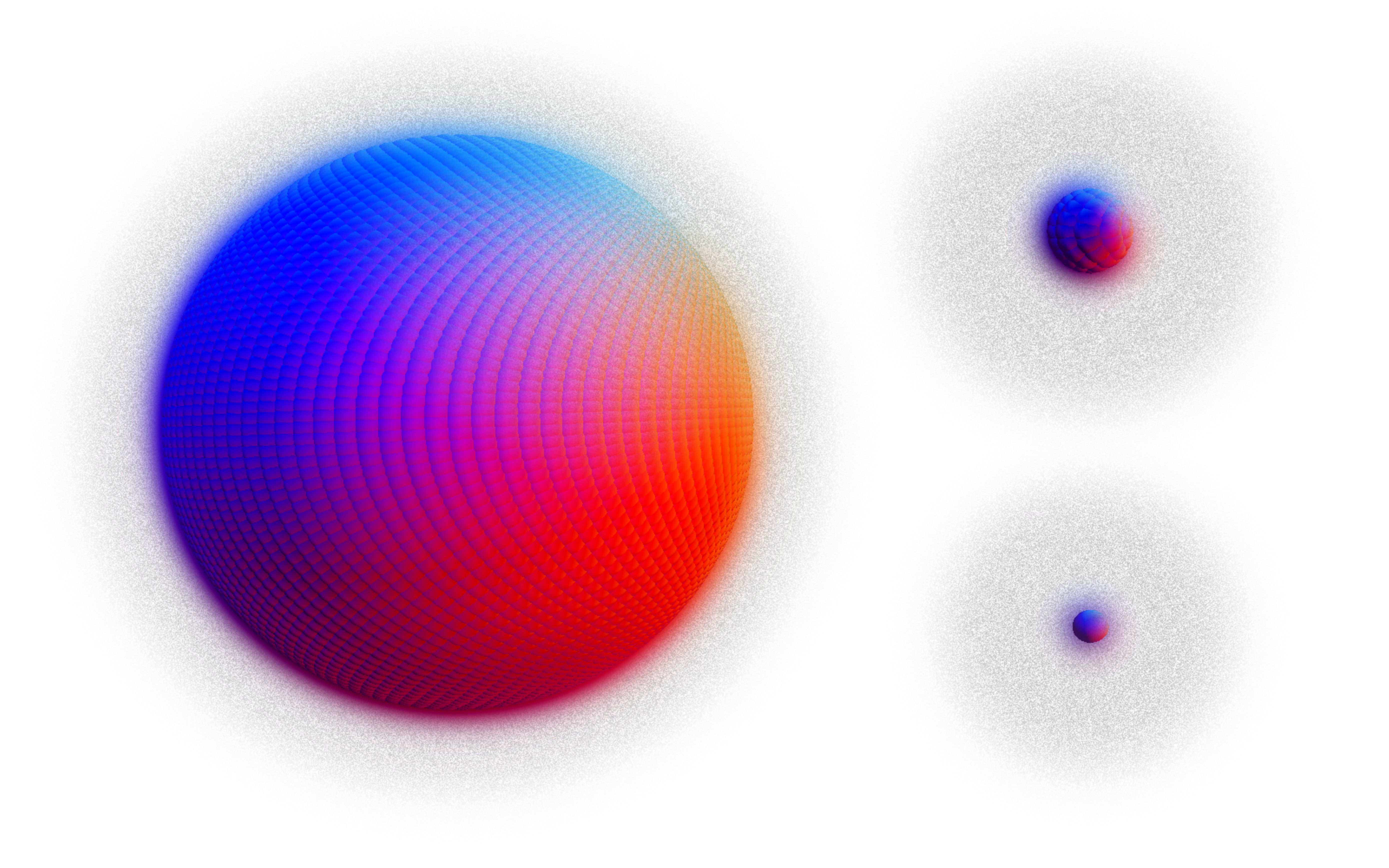}
\caption{\label{fig:scale}
Depiction of systems of decreasing particle number.
For systems consisting of a small number particles, the argument in Fig.~\ref{fig:boundary} obviously does not apply.
Furthermore, it is known that neglecting the interaction energy in these finite systems always results in infinite relaxation and thermalization times.}
\end{figure}

Finally, the term equilibrium is used in Ref.~\cite{Linden09} to describe what in our terminology are steady states and in Ref.~\cite{Reimann10} to describe what in our terminology are stationary states. Both cases have been covered in Refs.~\cite{Short11,Short12} with the single term equilibrium. These states do necessarily meet our more stringent criteria of equilibrium described above. Here we refer to the result of these works using the terminology we defined above. 

\subsection{Model and Assumptions}

We consider unitary dynamics of the closed system (C) described by the Hamiltonian $\mathbf{H}_\mathrm{C}$ consisting of the system of interest (S) and its environment (E) with interaction (I) 
\begin{align}
\label{eq:Hamiltonian}
\mathbf{H}_\mathrm{C} &= \mathbf{H}_\mathrm{S} + \mathbf{H}_\mathrm{E} + \mathbf{H}_\mathrm{I} + \mathbf{H}_\mathrm{R} \, ,
\end{align}
where $\mathbf{H}_\mathrm{R}$ contains all of the ``renormalization" (R) effects.
The interaction generates environmental correlation functions, c.f.
Eqs.~\eqref{eq:tensorInt}, \eqref{eq:alpha}),
and we assume these correlations to be decaying functions.
This assumption allows for irreversible dynamics in the open system.
Implicit in this assumption is that the environment contains a continuum of
modes (e.g. infinite volume).
This latter assumption can be satisfied by coupling the system directly
to field degrees of freedom that are uncountably infinite,
such as the electromagnetic field. Note however that we do not assume the interaction Hamiltonian to be negligible compared to the system Hamiltonian.

Finally, for mathematical simplicity we assume the initial state of the system
and environment to be uncorrelated at $t=0$\footnote{The implication of initial correlations are considered in Ref.~\cite{Tasaki07,Correlations}:
Correlated initial states are more physical, particularly in the early time evolution,
but they have essentially no bearing on the mathematical results we derive herein, which are focused upon the asymptotic time evolution.}
\begin{align}
\label{eq:product}
\boldsymbol{\rho}_\mathrm{C}(0) &= \boldsymbol{\rho}_\mathrm{S}(0) \otimes \frac{e^{-\beta \mathbf{H}_\mathrm{E}}}{Z_\mathrm{E}(\beta)} \, ,
\end{align}
where the environment (a thermal reservoir) is in its isolated equilibrium state with partition function $Z_\mathrm{E}(\beta)=\mathrm{Tr_E}[e^{-\beta \mathbf{H}_\mathrm{E}}]$, and the system (S) is in an arbitrary state.

The assumption of a thermal state for the environment can be justified, for
instance, by the approach of Popescu et al. \cite{Popescu06} in the weak-coupling limit, by giving the environment its own
environment, without any restriction on the system-environment
coupling strength.
In this sense the work of Popescu et al., and those prior, serve as a pedagogical
springboard for our analysis of strongly coupled systems.

\subsection{Results}

It is well known that in the limit of vanishing interaction strength, the open system relaxes to its thermal state \cite{vanHove55,Davies76,Geva00,QOS}
\begin{align}
\lim_{\gamma \to 0} \lim_{t \to \infty} \boldsymbol{\rho}_\mathrm{S}(t) &= \frac{e^{-\beta \mathbf{H}_\mathrm{S}}}{Z_\mathrm{S}(\beta)} \, , \label{eq:ES}
\end{align}
where $\boldsymbol{\rho}_\mathrm{S}(t) = \mathrm{Tr}_\mathrm{E} [\boldsymbol{\rho}_\mathrm{C}(t)]$ denotes the reduced density matrix and $\gamma$ a generic relaxation rate of the open system.
Note that all relaxation rates are, at minimum, second order in the interaction, being primarily determined by the two-time correlations of the environment.

In Ref.~\cite{Mori08}, it was shown to second-order in the interaction, and for a single tensor-product coupling of system and environment operators,
that the open system can be confirmed to relax to the reduced closed system thermal state
\begin{align}
\lim_{t \to \infty} \boldsymbol{\rho}_\mathrm{S}(t) &= \mathrm{Tr}_\mathrm{E}\!\left[ \frac{e^{-\beta \mathbf{H}_\mathrm{C}}}{Z_\mathrm{C}(\beta)} \right] \, . \label{eq:ER}
\end{align}
We extend this proof to general system-environment
couplings.
For zero-temperature environments we demonstrate agreement with the ground
state obtained from the time-independent Schr\"{o}dinger equation.
Moreover,
we give a non-perturbative proof of
Eq.~\eqref{eq:ER} for the exactly-solvable model of $N$-oscillator quantum
Brownian motion (N-QBM),
wherein the interacting system and environment are linear.
In that
model we are also able to rigorously prove that all multi-time correlations of the open system relax to those of the closed system thermal state with non-vanishing interaction.
Correspondence of the multi-time correlations is an important consideration as, outside of the Markovian regime, the dynamics of the multi-time correlations cannot be generated by the dynamics of the single-time correlations,
as per the quantum regression theorem (QRT) \cite{Swain81}.

\subsubsection*{The \emph{reduced}, closed system thermal state}

It is important to emphasize that Eq.~\eqref{eq:ER} pertains strictly to the open system S and not to the closed system (S +  E),
as equilibration requires not only a reservoir and late-time limit, but also a degree of coarse graining.
As we show in Sec.~\ref{sec:coarse}, if one considers any individual mode of the environment,
then its dependence upon the initial state of the system never decays.
In this sense, information pertaining to the system's past is encoded in the environment, but only when considering the state of the closed system (S + E).
However, upon coarse graining the environment by considering the time-evolution of a continuum of environment energies, and not one individual mode energy,
then all dependence upon the initial state of the system is seen to decay away in time.
In this sense, information pertaining to the system's past is only measurable for a finite span of time.

The above statement is based on the fact that, while the open system experiences irreversible dynamics: dissipation, diffusion and decoherence, the closed system (S + E) experiences reversible dynamics.
Consider, for instance, the coupling of a mixed state of the system to a zero-temperature reservoir.
Given unitary dynamics, the joint state of the system and environment cannot relax from a mixed state into a pure state (the ground state of the interacting theory).
However, the environment is exceedingly large when compared to the system, and so the system's entropy, when spread out over every mode of the environment, can become immeasurable.
This is a general phenomena of environmentally-induced irreversible dynamics: conserved quantities such as energy and entropy can flow into the environment, and owing to the overwhelmingly large number of degrees of freedom, become difficult to track or retrieve.

The paper is organized as follows: In Sec.~\ref{sec:NQBM} we derive the equilibrium state for the linear N-QBM model. In Sec.~\ref{sec:nonlinear} we extend our analysis to nonlinear systems via perturbation theory. In Sec.~\ref{sec:discussion} we summarize our results and compare them to relevant works in the literature and provide some new insights into the key issues. Some technical details have been provided and the notation is defined in the Appendices.

\section{Linear systems}
\label{sec:NQBM}
\subsection{The Lagrangian}

Our treatment of the N-QBM model is based on \cite{NQBM}. The model is that
of a continuous and linear system with finite and countable
degrees of freedom, with Lagrangian
$L_\mathrm{sys}(\mathbf{X},\dot{\mathbf{X}})$,
bilinearly coupled, via a Lagrangian $L_\mathrm{int}(\mathbf{X},\mathbf{x})$,
to a linear environment with an infinite (and possibly continuous) number of
degrees of freedom, with Lagrangian
$L_\mathrm{env}(\mathbf{x},\dot{\mathbf{x}})$.
\begin{align}
L =&\; L_\mathrm{sys}(\mathbf{X},\dot{\mathbf{X}}) +
L_\mathrm{env}(\mathbf{x},\dot{\mathbf{x}}) +
L_\mathrm{int}(\mathbf{X},\mathbf{x}) + L_\mathrm{ren}(\mathbf{X}) \, , \\
L =&\; \frac{1}{2} \left( \dot{\mathbf{X}}^\mathrm{T} \mathbf{M} \,
\dot{\mathbf{X}} - \mathbf{X}^\mathrm{T} \mathbf{C} \, \mathbf{X} \right)
+ \frac{1}{2} \left( \dot{\mathbf{x}}^\mathrm{T} \mathbf{m} \,
\dot{\mathbf{x}}  - \mathbf{x}^{\!\mathrm{T}} \mathbf{c} \, \mathbf{x}
\right) \nonumber \\
& -\mathbf{x}^{\!\mathrm{T}} \mathbf{g} \, \mathbf{X} +
L_\mathrm{ren}(\mathbf{X}) \, .
\end{align}
We assume that the spring constant matrices $\mathbf{C}, \mathbf{c}$ as well as
the mass matrices $\mathbf{M}, \mathbf{m}$ are real and positive definite, and
can be considered in general to be symmetric.
If necessary, one can relax the positivity condition and even consider
time-dependent mass matrices, spring constant matrices and system environment
coupling matrix $\mathbf{g}$ \cite{Fleming11T}.
Such a model environment can emulate any source of Gaussian noise with proper
choice of coupling.
To ensure that the free and interacting system are similar in behavior, we will
also include the ``renormalization" $L_\mathrm{ren}(\mathbf{X})$.
Our choice of ``renormalization" will be equivalent to inserting the entire
system-environment interaction in the square of the potential:
\begin{align}
L =&\; \frac{1}{2} \left( \dot{\mathbf{X}}^\mathrm{T} \mathbf{M} \,
\dot{\mathbf{X}} - \mathbf{X}^\mathrm{T} \mathbf{C} \, \mathbf{X} \right)
\nonumber \\
\label{eq:linearL}
& + \frac{1}{2} \left( \dot{\mathbf{x}}^\mathrm{T} \mathbf{m} \,
\dot{\mathbf{x}} - \left[ \mathbf{x} - \mathbf{c}^{-1} \mathbf{g} \, \mathbf{X}
\right]^\mathrm{T} \!\! \mathbf{c} \left[ \mathbf{x} - \mathbf{c}^{-1}
\mathbf{g} \, \mathbf{X} \right] \right) ,
\end{align}
since this keeps the phenomenological system-system couplings from changing.

\subsection{The Langevin Equation}
\label{sec:Langevin}

For the linear system there are several formalisms which produce the same
Langevin Equation.
The most direct is via integrating out environment degrees of freedom in the
Heisenberg equations of motion \cite{FordOconnell88} and then considering the
symmetrized moments.
Another is to consider the characteristic curves of the system + environment's
Fokker-Plank equation \cite{Fleming11T}.
Finally, one can integrate out both the environment degrees of freedom and the
relative system coordinate $\boldsymbol{\Delta} = \mathbf{X} - \mathbf{X}'$,
while leaving only the average system coordinate $\boldsymbol{\Sigma} =
(\mathbf{X} + \mathbf{X}')/2$,
in the double path integral of the reduced system propagator in the influence
functional formalism \cite{CRV03}.
In general (for nonlinear systems) there is no necessary correspondence between
these formalisms and only the first may be well defined, but here the Langevin
equation is simply
\begin{align}
\label{eq:langeq}
\mathbf{M}\, \ddot{\mathbf{X}}(t) + 2 \! \int_0^t \!\! \rd s \,
\boldsymbol{\gamma}(t,s) \, \dot{\mathbf{X}}(s) + \mathbf{C}\,
\mathbf{X}(t) &= \boldsymbol{\xi}(t) - 2 \, \boldsymbol{\gamma}(t) \, \mathbf{X}_0 \, ,
\end{align}
where $\bgamma$ is the damping kernel and $\bxi$ is the noise given by:
\begin{align}
\label{eq:dampkern}
\boldsymbol{\gamma}(t,s) &= +\mathbf{g}^\mathrm{T}
\mathbf{m}^{\!-\frac{1}{2}} \frac{
\cos\!\left(\boldsymbol{\omega}[t\!-\!s]\right) }{ 2 \,
\boldsymbol{\omega}^2 } \mathbf{m}^{\!-\frac{1}{2}} \mathbf{g} \, , \\
\label{eq:noise}
\bxi(t) &= \bgt \left( \bfdot(t) \, \blm \, \bxz + \blf(t) \, \bpz \, \right)\,
,\\
\label{eq:f(t)}
\blf(t) &= \bmtomoot \frac{\sinomega t)}{\bomega} \bmtomoot\, ,\\
\label{eq:omega}
\boldsymbol{\omega}^2 & \equiv \mathbf{m}^{\!-\frac{1}{2}} \mathbf{c} \,
\mathbf{m}^{\!-\frac{1}{2}} \, ,
\end{align}
where $\blf$ is the free Green's function of the reservoir positions and $\boldsymbol{\omega}$ is the free reservoir frequencies upon diagonalization.
Note that the damping kernel is independent of the environment's initial state, whereas the
properties of noise are determined by the environment's initial state.

We consider the case in which the system and environment are uncorrelated at $t=0$
and the environment is in its thermal state $e^{-\beta \bH_\rE}/Z_\rE(\beta)$.
The
noise has zero mean and the two time correlation is given by the noise
kernel
\begin{equation}
\bnu(t,t') = \lbr \boldsymbol{\xi}(t) \, \boldsymbol{\xi}(t')^\mathrm{T}
\right\rangle_{\boldsymbol{\xi}} \, ,
\end{equation}
where the Gaussian average over the stochastic process $\boldsymbol{\xi}$ is equivalent to tracing over the environment degrees of freedom.
The noise and damping kernels satisfy then the
fluctuation-dissipation relation
(here in the Fourier domain)
\begin{align}
\tilde{\boldsymbol{\nu}}(\omega) &= \tilde{\kappa}(\omega)  \,
\tilde{\boldsymbol{\gamma}}(\omega) \, , \label{eq:FDR} \\
\tilde{\kappa}(\omega) & \equiv \hbar \omega \, \coth\!\left( \frac{\hbar
\omega}{2 k_\mathrm{B} T} \right) \, , \label{eq:kappa}
\end{align}
with the Fourier transform defined
\begin{equation}
\tilde{f}(\omega) \equiv \int_{\!-\infty}^{+\infty} \!\!\! \rd t \, e^{-\imath
\omega t} \, f(t) \, ,
\end{equation}
and where $\tilde{\kappa}$ is the (quantum) FDR kernel.
Therefore, the problem is completely specified in terms of the damping kernel.

Given that our damping kernel is stationary, the Langevin equation can be expressed in the Laplace domain as
\begin{align}
\left[ z^2 \mathbf{M} + 2 z \hat{\boldsymbol{\gamma}}(z) + \mathbf{C} \right]
\hat{\mathbf{X}}(z) &= [ z \mathbf{M} \, \mathbf{X}_0 + \mathbf{P}_0 ] +
\hat{\boldsymbol{\xi}}(z)\, ,
\end{align}
where $\mathbf{P}=\bM\, \dot{\mathbf{X}}$ and $(\mathbf{X}_0, \mathbf{P}_0)$
correspond to the initial values at $t=0$,
and with the Laplace transform defined
\begin{equation}
\hat{f}(z) \equiv \int_0^\infty \!\! \rd t \, e^{-zt} \, f(t) \, .
\end{equation}
Formally, the solutions can be easily found by inversion:
\begin{align}
\hat{\mathbf{X}}(z) &= \hat{\mathbf{G}}(z) \, [ z \mathbf{M} \, \mathbf{X}_0 +
\mathbf{P}_0 ] + \hat{\mathbf{G}}(z) \, \hat{\boldsymbol{\xi}}(z)\, , \\
\label{eq:Ghat}
\hat{\mathbf{G}}(z) &= \left[ z^2 \mathbf{M} + 2 z \hat{\boldsymbol{\gamma}}(z)
+ \mathbf{C} \right]^{-1} \, .
\end{align}
Note that since our damping kernel is symmetric, i.e.
$\boldsymbol{\gamma}(t,s) = \boldsymbol{\gamma}(t,s)^\mathrm{T}$, the
same will be true for the propagator $\mathbf{G}(t,s)$ and its Laplace
transform.
It is also useful to consider the following representation:
\begin{align}
\hat{\mathbf{G}}(z) &= \mathbf{M}^{-\frac{1}{2}} \left[ z^2 + 2 z \,
\hat{\boldsymbol{\lambda}}(z) + \boldsymbol{\Omega}^2 \right]^{-1}
\mathbf{M}^{-\frac{1}{2}} \, , \label{eq:G(z)M} \\
\hat{\boldsymbol{\lambda}}(z) & \equiv \mathbf{M}^{-\frac{1}{2}}
\hat{\boldsymbol{\gamma}}(z) \, \mathbf{M}^{-\frac{1}{2}} \, , \\
\boldsymbol{\Omega}^2 & \equiv \mathbf{M}^{-\frac{1}{2}} \mathbf{C} \,
\mathbf{M}^{-\frac{1}{2}} \, ,
\end{align}
where the eigenvalues of $\boldsymbol{\Omega}^2$ coincide with the squared
frequencies of the normal modes of the free system.
Back in the time domain we have
\begin{align}
\label{eq:X(t)}
\mathbf{X}(t) &= \dot{\mathbf{G}}(t) \, \mathbf{M} \, \mathbf{X}_0 +
\mathbf{G}(t) \,\mathbf{P}_0 + (\mathbf{G} * \boldsymbol{\xi})(t)\, ,
\end{align}
with $*$ denoting the Laplace convolution, defined as
\begin{equation}
(A*B)(t) = \int_0^t \!\! \rd s \, A(t\!-\!s) \, B(s).
\end{equation}

For more general Gaussian states, for which the system and environment are
correlated, the noise can be correlated with $(\mathbf{X}_0, \mathbf{P}_0)$
and the noise kernel modified.
This is the case for the closed system thermal state
given by the density matrix $e^{-\beta \bH_\rC}/Z_\rC(\beta)$ which we
investigate below.

\subsection{Single-time correlations in the closed system thermal state}
\label{sec:singletimecorrelations}

In this section we calculate the single-time correlations in the closed system
thermal state of the N-QBM model. The partition function for the N-QBM
model has
been derived in App.~\ref{app:partitionfunction}, Eq.~\eqref{eq:PF}. In the
rest of the paper including the appendices we suppress the dependence of the
paritition function on $\beta$ for brevity of notation.
As a
first step we take the logarithm of the
partition function and write it as:
\begin{align}
\label{eq:logPF}
 \log Z_\rC 
&= \log Z_\rE -\frac{1}{2} \tr \log \mathbf{M}^{-1} -\frac{1}{2} \tr\log \mathbf{C} \nonumber \\
& -\sum_{r=1}^\infty \tr \log \left( \mathbf{M}^{-1}\hat{\mathbf{G}}(\nur)^{-1} \right)+\mathrm{constant}
\end{align}

We begin by making a general observation. Consider the thermal state of a
system described by a Hamiltonian where the momenta appear only in the kinetic
energy term of the form $\sum_a p_a^2/2m$. Then all correlations between
position and momentum operators vanish: $\left\langle x_a p_b
\right\rangle=0$. 
This can be seen by noting that all correlations are
time-translation-invariant in equilibrium and forming the derivatives $\frac{\rd}{\rd
t} \left\langle x_a(t) x_b(t) \right\rangle$ and $\frac{\rd}{\rd
(t-t')}\left\langle x_a(t) x_b(t') \right\rangle\big|_{t=t'}$.
This observation applies to N-QBM model.

Let angular bracket with the subscript C denote expectation values in
the closed system thermal state. Expectation values corresponding to the uncorrelated initial state are
denoted by
attaching the subscript E
to the bracket.
For the purpose of partial differentiation, the partition function is to
be regarded as a function of $\bC,\ \bM,\ \bc,\ \blm,\ \bg$ and not (explicitly) of
$\bomega$.
With a straight-forward application of Theorem~\ref{thm:difftrick}, the
reduced system correlations are given by:
\begin{align}
\label{eq:XXform}
 \lbr \mathbf{X} \, \mathbf{X}^{\rT}\rbr&=-\frac{2}{\beta}\frac{\partial
\log Z_\rC}{\partial \mathbf{C}} \, ,\\
 \lbr \mathbf{X} \, \mathbf{P}^{\rT}\rbr&= \left \langle \mathbf{P} \,
\mathbf{X}^{\rT}\right \rangle=0 \, ,\\
\label{eq:PPform}
 \lbr \mathbf{P} \, \mathbf{P}^{\rT}\rbr&=-\frac{2}{\beta}\frac{\partial
\log Z_\rC}{\partial \mathbf{M}^{-1}} \, .
\end{align}

The position-position and position-momentum correlations between system and
reservoir modes are calculated similarly:
\begin{align}
\label{eq:Xxform}
\lbr \mathbf{X} \, \mathbf{x}^{\rT}\rbr &=  \lbr \mathbf{x} \, \mathbf{X}^{\rT}\rbr^{\rT} \\
& =\frac{1}{\beta}\frac{\partial \log Z_C}{\partial \mathbf{g}^{\rT}} +\lbr \mathbf{X} \, \mathbf{X}^{\rT}\rbr \, \mathbf{g}^{\rT} \mathbf{c}^{-1} \, , \nonumber \\
\lbr \mathbf{X} \, \mathbf{p}^{\rT}\rbr&= \lbr \mathbf{p} \, \mathbf{X}^{\rT}\rbr^{\rT}=0 \, ,  \\
\lbr \mathbf{P} \, \mathbf{x}^{\rT}\rbr&=\lbr \mathbf{x} \, \mathbf{P}^{\rT}\rbr^{\rT}=0 \, .
\end{align}
To calculate the momentum-momentum correlations between system and environment
we take the time derivative of $\lbr \mathbf{X}(t)
\mathbf{p}^{\rT}(t)\rbr$ and set it to zero. Since in the closed system thermal state all
expectation values are time-independent, we know that there is in fact no
dependence on time. Using the equations of motion it is straight-forward to
show that:
\begin{align}
\label{eq:Ppform}
\lbr \mathbf{P} \, \mathbf{p}^{\rT}\rbr
&=\mathbf{M} \, \lbr \mathbf{X} \, \mathbf{x}^{\rT}\rbr \, \mathbf{c}
-\mathbf{M} \, \lbr\mathbf{X} \, \mathbf{X}^{\rT}\rbr \, \mathbf{g}^{\rT} \, .
\end{align}
The environment correlations can 
be calculated by direct
differentiation of the partition function:
\begin{align}
 \label{eq:xxform}
 \lbr \mathbf{x} \,
\mathbf{x}^{\rT}\rbr
&=-\frac{2}{\beta}\frac{\partial \log Z_C}{\partial \mathbf{c}}
+\mathbf{c}^{-1}\mathbf{g} \, \lbr\mathbf{X} \, \mathbf{X}^{\rT}\rbr \,
\mathbf{g}^{\rT}\mathbf{c}^{-1} \, ,\\
\lbr \mathbf{x} \, \mathbf{p}^{\rT}\rbr&= \lbr \mathbf{p} \,
\mathbf{x}^{\rT}\rbr=0 \, ,\\
\label{eq:ppform}
\lbr \mathbf{p} \, \mathbf{p}^{\rT}\rbr&=-\frac{2}{\beta}\frac{\partial \log
Z_C}{\partial \mathbf{m}^{-1}} \, .
\end{align}

Now we are in a position to determine all the single-time correlations of the interacting theory in the closed system thermal state.
Since the equilibrium state is stationary these single-time correlations are time-independent.
The details for some of these formulae are provided in App~\ref{app:singletime}. All the nonzero correlations are given by:
\begin{align}
\label{eq:<XX>} \lbr \mathbf{X} \, \mathbf{X}^{\rT}\rbr &=\frac{1}{\beta} \hat{\mathbf{G}}(\nu_0)+\frac{2}{\beta} \sum_{r=1}^\infty\hat{\mathbf{G}}(\nur) \, ,\\
\label{eq:<PP>} \lbr \bP \, \bPt \rbr &= \frac{1}{\beta} \left( \bM - \nu_0^2 \bM \, \bGhat(\nu_0) \, \bM \right) \\
& + \frac{2}{\beta}\rsum \left( \bM -  \nur^2 \bM \, \bGhatnur \, \bM \right) \, , \nonumber\\
\label{eq:<xX>} \lbr \mathbf{X} \, \mathbf{x}^{\rT}\rbr&=\lbr \mathbf{X} \, \mathbf{X}^{\rT}\rbr \mathbf{g}^{\rT}\mathbf{c}^{-1}
-\frac{2}{\beta}\sum_{r=1}^\infty \nur \hat{\mathbf{G}}(\nur) \, \hat{\boldsymbol{\gamma}}(\nur) \, \mathbf{g}^{-1} \, , \\
\label{eq:<pP>} \lbr \mathbf{P} \, \mathbf{p}^{\rT}\rbr &=\mathbf{M}\lbr \mathbf{X} \, \mathbf{x}^{\rT}\rbr\mathbf{c}-\mathbf{M}\lbr\mathbf{X} \mathbf{X}^{\rT}\rbr \mathbf{g}^{\rT} \, , \\
\label{eq:<pp>} \lbr \bpz \, \bptz \rbr &= \left \langle \bpz \, \bptz \right \rangle_\rE \\
&-\frac{2}{\beta}  \rsum \nur^2 \blm \, \bfhatnur \, \bg \, \bGhatnur \, \bgt \bfhatnur \, \blm \, , \nonumber \\
\label{eq:<xx>} \lbr \bxz \, \bxtz \rbr &= \left \langle \bxz \, \bxtz \right \rangle_\rE + \bctomo \bg \lbr \bXz \, \bX \rbr \bgt \bctomo \\
\nonumber &-\frac{2}{\beta} \rsum  \nur^2 \left( \bctomo \blm \, \bfhatnur \, \bg \, \bGhatnur \, \bgt \bfhatnur \right) \\
\nonumber &-\frac{2}{\beta} \rsum \nur^2 \left( \bfhatnur \, \bg \,  \bGhatnur \, \bgt \bfhatnur \, \blm \, \bctomo \right) \\
\nonumber &-\frac{2}{\beta} \rsum \nur^4  \left( \bctomo \blm \, \bfhatnur \, \bg \, \bGhatnur \, \bgt \bfhatnur \, \blm \, \bctomo \right)  ,
\end{align}
where $\bfhat$ is the Laplace transform of the free reservoir propagator given by Eq.~\eqref{eq:f(t)}
and $\nur = 2\pi \mathrm{r} / \hbar \beta$ are the Matsubara frequencies.

\subsection{Equivalence of single-time correlations for the open system}
\label{sec:singletimeequivalence}

In this subsection we show that the single-time correlations of
system variables for the uncorrelated initial state are asymptotically
identical to the
single-time correlations corresponding to the closed system thermal state. We
start
by calculating the variances for the closed system thermal state.
The requirement that $\bG(t)$ is a decaying function means that the Laplace
transform $\bGhat(z)$ is analytic in the right half-plane. Hence
$\bGhat(-\imath \omega)$ is analytic in the upper-half plane. On the other
hand $\coth(\beta \hbar \omega/2)$ has simple poles on the imaginary axis at
the Matzubara
frequencies $\nur$. The summations over $r$ in Eq.~\eqref{eq:<XX>} can be
written as a contour integral using Cauchy's theorem:
\begin{align}
 \lbr \bX \, \bXt \rbr &= \frac{\beta \hbar/2}{2 \pi \imath}\times
\frac{2}{\beta} \int_C \rd z \coth(\beta \hbar z/2) \, \bGhat(-\imath z) \, .
\end{align}
The contour of integration is chosen to encircle the upper-half plane in a
counter-clockwise direction. The poles on the imaginary axis at Matzubara
frequencies $\nur$ for
$r\geq 1$ are encircled, but only half of the pole at the origin is enclosed.
The arc of the contour does not contribute to the integral when the radius is
taken to infinity. Hence we can write this expression as an integral on the
real line.
Furthermore, by the symmetry of the integrand, the real part vanishes and the integral is given by:
\begin{align}
\label{eq:<X^2>}
  \lbr \bX \, \bXt \rbr &= \frac{\hbar}{2 \pi} \int_{\!-\infty}^{+\infty} \!\!\! \rd \omega \coth(\beta \hbar \omega/2) \, \mathrm{Im}\!\left[ \bGhat(-\imath \omega) \right] .
\end{align}
A similar argument can be used to derive:
\begin{align}
\label{eq:<P^2>}
 \lbr \bP \, \bPt \rbr &= \frac{\hbar}{2\pi} \int_{\!-\infty}^{+\infty} \!\!\! \rd \omega \, \omega^2 \coth(\beta \hbar \omega/2) \, \mathrm{Im} \! \left[ \bGhat(-\imath \omega) \right] .
\end{align}
Eqs.~(\ref{eq:<X^2>},\ref{eq:<P^2>}) are identical to the results obtained
by \cite{NQBM} for the asymptotic values of variances corresponding to an uncorrelated initial state.
Therefore we have proven that the single-time correlations of the open system relax to those of the closed system thermal state.

\subsection{Equivalence of multi-time correlations}
\label{sec:generalequivalence}

In this section we generalize the results of the previous section to include
multi-time correlations.
We begin by calculating the two-time correlation function  $\lbr \bX(t) \,
\bX(t')^\mathrm{T} \rbr$ using the trajectories obtained from the Langevin
equation.
Note that for the closed system thermal state this quantity is stationary. To simplify the proof we make
use of this observation and take the late-time limit of the closed system thermal
state as well without loss of generality. This trick makes the comparison of
the two cases easier and reduces the amount of computation.

The dynamics of the system is given by the solution \eqref{eq:X(t)} of the
Langevin equation which is valid for any initial state. The dependence
on
initial state is hidden in the correlations between $\bXz,\ \bPz$ and
$\bxi(t)$. The
two-time position correlation is given by
\begin{widetext}
\begin{align}
\nonumber
\lbr \bX(t) \, \bX(t')^\mathrm{T} \rbr &= \bGdot(t) \, \bM \lbr \bXz \, \bXtz \rbr \bM \, \bGdot(t')+ \bG(t) \lbr \bPz \, \bPtz \rbr \bGt(t') \\
&\nonumber + \bGdot(t) \, \bM \int_0^{t'} \!\! \rd s' \, \lbr \bXz \,
\bxi(s')^\rT \rbr \bG(t' \!-\! s')
+\int_0^{t} \!\! \rd s \, \bG(t \!-\! s) \lbr \bxi(s) \, \bXtz \rbr \bM
\, \bGdot(t') \\
&\nonumber + \bG(t) \int_0^{t'} \!\! \rd s' \, \lbr \bPz \, \bxi(s')^\rT \rbr
\bG(t' \!-\! s')
+\int_0^{t} \!\! \rd s \, \bG(t \!-\! s) \lbr \bxi(s) \, \bPtz \rbr
\bGdot(t') \\
\label{eq:<X(t)X(t')>} &+ \int_0^t \!\! \rd s \int_0^{t'} \!\! \rd s' \,
\bG(t \!-\! s) \lbr \bxi(s) \, \bxi(s')^\rT \rbr \bG(t' \!-\! s') \, .
\end{align}
\end{widetext}
As mentioned earlier unlike the uncorrelated initial state the terms in the
second and third lines do not vanish in the closed system thermal state.
We consider the case where
\begin{align}
 \lim_{t\rightarrow \infty} \bG(t)=
\lim_{t\rightarrow \infty} \bgamma(t)=0 \, .
\end{align}
This is the criteria for dissipative dynamics. Under these assumptions the first two terms in
Eq.~\eqref{eq:<X(t)X(t')>} vanish in the late-time limit for any initial
state.
The terms in the second and third lines have one factor of $\bG(t)$ or
$\bGdot(t)$ that goes to zero in the late-time limit multiplied by a
convolution integral. In App.~\ref{app:singletime} we show that these
convolution integrals are finite. Hence the terms in
second and third lines also vanish asymptotically. Finally we show the
equivalence of the term in the last line for the uncorrelated and thermal
initial states at late times in App.~\ref{app:multitime}.

The comparison of more general multi-time correlations can be done
similarly using the trajectories of the Langevin equation. The above example
demonstrates how in the late-time limit the effects of initial conditions of
the system die out and the noise statistics of both preparations converge. The
equivalence at the level of trajectories ensures that all the multi-time
correlations will be identical.

Let us reiterate the result we just obtained: a linear system
linearly coupled to a linear thermal reservoir (with uncountably many degrees of
freedom) at inverse temperature $\beta$ does relax to the equilibrium state
described by \eqref{eq:ER}. This state is different from the Boltzmann state
given by \eqref{eq:ES} whenever the interaction between the system and
environment is not negligible. Moreover the
multi-time correlations of system observables also relax to their corresponding
values in the closed system thermal state.

\subsection{The effect of coarse graining}\label{sec:coarse}

Up until this point we only focused on the system degrees of freedom. Now we
turn our attention to the environment.
Following Ref.~\cite{NQBM,Fleming11T}, the trajectories of the environment
oscillators, as driven by the system oscillators, are given by
\begin{align}
\mathbf{x}(t) &= \left[ \dot{\mathbf{f}}(t) \, \mathbf{m} \, \mathbf{x}(0) + \mathbf{f}(t) \, \mathbf{p}(0) \right] + \mathbf{f}(t) * \mathbf{g} \, \mathbf{X}(t) \, , \label{eq:x(t)}
\end{align}
in terms of their free propagator $\mathbf{f}(t)$ and frequency matrix
$\boldsymbol{\omega}$ given by Eqs.~(\ref{eq:f(t)},\ref{eq:omega}).
Into Eq.~\eqref{eq:x(t)} we substitute the system trajectories, which are damped oscillations driven by noise for the continuum environment:
\begin{align}
\mathbf{X}(t) &= \left[ \dot{\mathbf{G}}(t) \, \mathbf{M} \, \mathbf{X}(0) + \mathbf{G}(t) \, \mathbf{P}(0) \right] + \mathbf{G}(t) * \boldsymbol{\xi}(t) \, .
\end{align}
We then find the environmental dependence upon the initial state of the system to be
\begin{align}
\mathbf{x}(t) &= \mathbf{f}(t) \, \mathbf{g} * \left[ \dot{\mathbf{G}}(t) \, \mathbf{M} \, \mathbf{X}(0) + \mathbf{G}(t) \, \mathbf{P}(0) \right] + \cdots \, ,
\end{align}
with all additional terms only dependent upon the initial state of the environment.
The system-dependent terms correspond to a convolution of harmonic oscillations of the environment with non-locally damped oscillations of the system.
Resolving these integrals leads to some terms which oscillate with environment frequencies $\omega$ and do not decay.

As a simple example, consider the local damping of a single system oscillator.
The open-system propagator or Green's function is given by
\begin{align}
G(t) &= \frac{\sin(\tilde{\Omega} t)}{M \tilde{\Omega}} e^{-\gamma_0 t} \, , \\
\tilde{\Omega} &= \sqrt{\Omega^2 - \gamma_0^2} \, .
\end{align}
The environment's dependence upon the initial state of the system is given by
\begin{align}
x_k(t) =& \left\{ X(0) \frac{\rd}{\rd t} + \frac{P(0)}{M} \right\} \left\{
\frac{\rd^2}{\rd t^2} - 2 \gamma_0 \frac{\rd}{\rd t} + \Omega^2 \right\} h_k(t)
\nonumber \\
& + \, \cdots \, , \\
h_k(t) \equiv& \frac{g_k f_k(t) }{\left(\omega_k^2 - \Omega^2\right)^2 + 4 \gamma_0^2 \omega_k^2} \, ,
\end{align}
plus terms that decay exponentially and the terms which depend upon the initial state of the environment.
The function $h_k(t)$ oscillates forever, the same as $f_k(t)$, and therefore the environment retains information pertaining to the initial state of the system forever.
However, this information is not measurable forever.
The system only interacts with the integrated trajectories, which resolve to a convolution of the damping kernel and open-system propagators.
\begin{align}
\mathbf{g}^\mathrm{T} \mathbf{x}(t) &= - 2 \, \dot{\boldsymbol{\gamma}}(t) * \left[ \dot{\mathbf{G}}(t) \, \mathbf{M} \, \mathbf{X}(0) + \mathbf{G}(t) \, \mathbf{P}(0) \right] + \cdots \, ,
\end{align}
and upon integrating over a continuum of environment frequencies (here performed by multiplication with the infinite matrix $\mathbf{g}^\mathrm{T}$)
the oscillatory terms decay in time.
Thus the late-time limit and coarse graining together are responsible for the erasure of all information pertaining to the initial state of the system.

\section{General systems}\label{sec:nonlinear}
\subsection{Steady state}
The time-evolution of the reduced density matrix of the open system can be generated by a perturbative master equation
\begin{align}
\dot{\boldsymbol{\rho}}_\mathrm{S}(t) &= \boldsymbol{\mathcal{L}}(t) \! \left\{ \boldsymbol{\rho}_\mathrm{S}(t) \right\}  ,
\end{align}
where the Liouville operator can be expanded in terms of the interaction Hamiltonian by a variety of methods \cite{Kampen97,Breuer03,Strunz04,QOS}.
\begin{align}
\boldsymbol{\mathcal{L}}(t) &= \boldsymbol{\mathcal{L}}_0 + \boldsymbol{\mathcal{L}}_1(t) + \boldsymbol{\mathcal{L}}_2(t) + \cdots \, , \\
\boldsymbol{\mathcal{L}}_0 \{ \boldsymbol{\rho}  \}&= -\imath [ \mathbf{H}_\mathrm{S} , \boldsymbol{\rho} ] \, ,
\end{align}
In general, $\boldsymbol{\mathcal{L}}_1(t)$ can be absorbed into the system Hamiltonian and so we will primarily concern ourselves with the second-order term.
For simplicity we will assume there is no degeneracy or near-degeneracy in the system energy spectrum.

Expanding the interaction Hamiltonian in terms of system $\mathbf{L}_n$ and environment $\mathbf{l}_n$ operators
\begin{align}
\mathbf{H}_\mathrm{I} &= \sum_n \mathbf{L}_n \otimes \mathbf{l}_n \, , \label{eq:tensorInt}
\end{align}
the multivariate master equation can be represented \cite{QOS}
\begin{equation}
\boldsymbol{\mathcal{L}}_2 \, \boldsymbol{\rho} = \sum_{nm} \left[ \mathbf{L}_n, \boldsymbol{\rho} \, (\mathbf{A}_{nm}\! \diamond \mathbf{L}_m)^\dagger - (\mathbf{A}_{nm}\! \diamond \mathbf{L}_m) \, \boldsymbol{\rho} \right] \, , \label{eq:WCGME}
\end{equation}
where the $\mathbf{A}$ operators and $\diamond$ product define the second-order operators
\begin{equation}
(\mathbf{A}_{nm}\! \diamond \mathbf{L}_m)(t) \equiv \int_0^t \!\! \rd s \,
\alpha_{nm}(t,s) \, \left\{ \boldsymbol{\mathcal{G}}_0(t,s) \,
\mathbf{L}_m(s) \right\} \, , \label{eq:WCOG}
\end{equation}
in terms of the zeroth-order (state) propagator of the system
\begin{align}
\boldsymbol{\mathcal{G}}_0(t,s) \{ \boldsymbol{\rho} \} &= e^{-\imath (t-s)
\mathbf{H}_\mathrm{S}} \, \boldsymbol{\rho} \, e^{+\imath (t-s)
\mathbf{H}_\mathrm{S}} \, ,
\end{align}
and the (multivariate) environmental correlation function
\begin{eqnarray}
\alpha_{nm}(t,s) & \equiv & \left \langle \underline{\mathbf{l}}_n(t) \,
\underline{\mathbf{l}}_m(s) \right \rangle_{\mathrm{E}} \, .
\label{eq:alpha}
\end{eqnarray}
The second-order operator can be expressed as the Hadamard product
\begin{equation}
\bra{\omega_i} \mathbf{A}_{nm}\! \diamond \mathbf{L}_m \ket{\omega_{i'}} = A(\omega_i \!-\! \omega_{i'}) \, \bra{\omega_i} \mathbf{L}_m \ket{\omega_{i'}} \, , \label{eq:Hadamard1}
\end{equation}
and, in the late-time limit, the second-order coefficients resolve
\begin{align}
A_{nm}(\omega) &= \frac{1}{2} \tilde{\alpha}_{nm}(\omega) - \imath \, \mathcal{P}\!\left[\frac{1}{\omega}\right] * \tilde{\alpha}_{nm}(\omega) \, , \label{eq:Hadamard2}
\end{align}
where $\tilde{\boldsymbol{\alpha}}(\omega)$ denotes the Fourier transform of
the stationary environment correlation function $\boldsymbol{\alpha}(t - s) =
\boldsymbol{\alpha}(t,s)$, $\mathcal{P}$ the Cauchy principal value and $*$
the appropriate Fourier convolution.

With the multivariate master equation detailed, we can prove relation \eqref{eq:ER} to second order in the interaction.
This generalizes the univariate proof in Ref.~\cite{Mori08}, which considered a single tensor-product interaction between the system and environment.
As the proof is straightforward in either case, we will give an outline and focus upon differences which arise in the multivariate treatment.

We are looking for the stationary state $\boldsymbol{\rho}_\beta$, such that
\begin{align}
\boldsymbol{\mathcal{L}} \{ \boldsymbol{\rho}_\beta \} &= 0 ,
\end{align}
we know from detailed balance that the zeroth-order stationary state is the
thermal state \eqref{eq:ES}, e.g. see \cite{QOS}.
Second-order corrections can be generated from the second-order master equation via canonical perturbation theory.
More explicitly, we have
\begin{equation}
\bra{\omega_i} \boldsymbol{\rho}_{\!\beta} \ket{\omega_j}_{i \neq j} \propto e^{-\beta \, \omega_{i}} \, \delta_{ij} - \imath \frac{ \bra{\omega_i} \boldsymbol{\mathcal{L}}_2 \{ e^{-\beta \, \mathbf{H}} \} \ket{\omega_j} }{\omega_i - \omega_j} \label{eq:pthermME} \, ,
\end{equation}
but only for the denoted off-diagonal perturbative corrections (in the energy basis $\ket{\omega}$).
As explained in Ref.~\cite{Accuracy}, due to unavoidable degeneracy,
specifically that the diagonal elements are all stationary to zeroth-order,
the second-order master equation cannot determine the second-order corrections to the diagonal elements of the density matrix.
Fortunately this does not greatly hamper our proof of correspondence.
By a simple application of the multivariate master equation to Eq.~\eqref{eq:pthermME}, we easily obtain these second-order corrections to the thermal state of the system.

\subsection{Equilibrium state}
We wish to compare the straightforward expansion of \eqref{eq:pthermME} to the reduced closed system thermal state at second order,
and so we require a perturbative expansion of \eqref{eq:ER}.
There exists such a perturbative expansion of exponential matrices utilizing the identity
\begin{equation}
\frac{\rd}{\rd\epsilon} e^{\mathbf{A} + \epsilon \, \mathbf{B}} = e^{\mathbf{A}
+ \epsilon \, \mathbf{B}} \! \int_0^1 \!\! \rd u \, e^{-u(\mathbf{A} +
\epsilon \, \mathbf{B})} \, \mathbf{B} \, e^{+u(\mathbf{A} + \epsilon \,
\mathbf{B})}  \, ,
\end{equation}
to obtain an operator-Taylor series in the perturbation $\epsilon \, \mathbf{B}$.
After a fair amount of simplification, one can determine the second-order
stationary state to be
\begin{align}
\boldsymbol{\rho}_{\!\beta} \propto & \; e^{-\beta \, \mathbf{H}_\mathrm{S}} \label{eq:pthermB} \\
& + e^{-\beta \, \mathbf{H}_\mathrm{S}} \! \int_0^\beta \!\! \rd\beta' \! \int_0^{\beta'}
\!\!\! \rd\beta'' \left \langle \underline{\mathbf{H}}_\mathrm{I}(-\imath
\beta')
\, \underline{\mathbf{H}}_\mathrm{I}(-\imath \beta'')
\right \rangle_{\mathrm{E}} \, , \nonumber
\end{align}
in terms of the complex-time operators
\begin{equation}
\underline{\mathbf{H}}_\mathrm{I}(-\imath \beta) \equiv e^{+\beta(\mathbf{H} + \mathbf{H}_\mathrm{E})} \, \mathbf{H}_\mathrm{I} \, e^{-\beta(\mathbf{H} + \mathbf{H}_\mathrm{E})} \, ,
\end{equation}
where the noise average is taken with respect to the free thermal state of the environment and
factors inside the environmental trace have been written to suggest their correspondence with the environmental correlation function evaluated at imaginary times.

The double integrals in Eq.~\eqref{eq:pthermB} reduce to
\begin{equation}
\sum_{nm} \int_0^\beta \!\!\! \rd\beta' \! \int_0^{\beta'} \!\!\! \rd\beta'' \,
\alpha_{nm}(-\imath \beta',-\imath \beta'') \, \underline{\mathbf{L}}_n(-\imath
\beta') \, \underline{\mathbf{L}}_m(-\imath \beta'') \, , \label{eq:pthermSM2}
\end{equation}
in terms of the complex-time operators
\begin{equation}
\underline{\mathbf{L}}(-\imath \beta) \equiv e^{+\beta\, \mathbf{H}_\mathrm{S}} \, \mathbf{L} \, e^{-\beta\, \mathbf{H}_\mathrm{S}} \, .
\end{equation}
After a Fourier expansion of the complex-time correlation functions, expressions \eqref{eq:pthermME} and \eqref{eq:pthermSM2} can be compared term-by-term in the energy basis wherein the imaginary-time integrals of Eq.~\eqref{eq:pthermSM2} can be resolved as the master equation operators were.
Though the two expressions will then be composed of the same objects, they will not immediately appear to be equivalent.
The final step is to apply the relevant multivariate Kubo-Martin-Schwinger (KMS) relations (also found in \cite{QOS})
\begin{equation}
\tilde{\boldsymbol{\alpha}}(+\omega) = \tilde{\boldsymbol{\alpha}}^{\!\mathrm{T}\!}(-\omega) \, e^{-\frac{\omega}{T}} = \tilde{\boldsymbol{\alpha}}^{*}(-\omega) \, e^{-\frac{\omega}{T}} \, , \label{eq:thermalsym}
\end{equation}
and then one can see that the two expressions are equivalent,
even in their missing diagonal perturbations.

As previously discussed, Eq.~\eqref{eq:pthermME} is missing diagonal perturbations due to degeneracies inherent in all perturbative master equations.
The same discrepancy in Eq.~\eqref{eq:pthermSM2} stems from the perturbative expansion of the equilibrium state being inherently secular in $\beta$.
Both expressions are equivalent and quite conveniently they both lack the same second-order corrections to the diagonal entries of the density matrix.
Therefore, as far as the second-order {dynamics} is concerned, our proof of correspondence is complete.

\subsection{Zero-Temperature Analysis}

Though correspondence was established,
the previous analysis was seen to be insufficient for calculation of low-temperature equilibrium states of the open system.
However, as we shall now show, at least for zero-temperature noise, it is still possible to easily construct the reduced closed system thermal states in terms of the same environmental correlation functions which occurred in the previous analysis.
The following relations were applied towards the inspection of two-level atoms
interacting via a zero-temperature quantum field in \cite{Dipole}.

In the zero-temperature regime we can apply mundane perturbation theory to
derive the stationary-state perturbations.
One merely considers the perturbed ground state of the system + environment
\begin{eqnarray}
\boldsymbol{\psi} &=& \boldsymbol{\psi}_0  + \boldsymbol{\psi}_1 + \boldsymbol{\psi}_2 + \cdots \, , \\
\boldsymbol{\psi}_0 & \equiv & \ket{0} \otimes \ket{0}_\mathrm{E} \, ,
\end{eqnarray}
and then traces out the environment
\begin{equation}
\boldsymbol{\rho}_{\beta} = \ket{0}\!\!\bra{0} + \left \langle
\boldsymbol{\psi}_2 \, \boldsymbol{\psi}_0^\dagger + \boldsymbol{\psi}_1 \,
\boldsymbol{\psi}_1^\dagger + \boldsymbol{\psi}_0 \,
\boldsymbol{\psi}_2^\dagger \right\rangle_{\!\mathrm{E}} + \cdots \, ,
\end{equation}
where we neglect the first moment of the reservoir as previously discussed.
Without loss of generality let us set the ground-state energy of the system to zero.
The second-order reduced ground state of ordinary perturbation theory provides
the following stationary state perturbation
\begin{equation}
\bra{\omega_i} \boldsymbol{\delta\!\rho}_{\!\beta} \ket{\omega_j} = \sum_{nmk} \frac{C_{ijk}^{nm}}{Z_0(\beta)} \bra{\omega_i} \mathbf{L}_m \ket{\omega_k} \bra{\omega_k} \mathbf{L}_n \ket{\omega_j} \, , \label{eq:G(beta)}
\end{equation}
where $Z_0(\beta)$ is the partition function of the free system and
with the off-diagonal (and non-resonant) coefficients given by
\begin{align}
& \left. C_{ijk}^{nm} \right|_{\omega_i \neq \omega_j} = + \mathrm{An}\!\left[ e^{-\beta \omega_k} \frac{A_{nm}(\omega_{ik})-A_{nm}(\omega_{jk})}{\omega_i - \omega_j} \right] \nonumber \\
& + \mathrm{An}\!\left[ \frac{e^{-\beta \omega_i} A_{mn}(\omega_{ki}) - e^{-\beta \omega_j} A_{mn}(\omega_{kj})}{\omega_i-\omega_j} \right]  ,
\end{align}
where $\omega_{ij} = \omega_i - \omega_j$ and $A_{nm}(\omega)$ are the second-order master equation coefficients in \eqref{eq:Hadamard2}.
`An' denotes the anti-Hermitian part; the Hermitian and anti-Hermitian parts are defined
\begin{align}
\mathrm{He}[ Q_{nm} ] &\equiv \frac{1}{2} ( Q_{nm} + Q_{mn}^* ) \, , \\
\mathrm{An}[ Q_{nm} ] &\equiv \frac{1}{2} ( Q_{nm} - Q_{mn}^* ) \, ,
\end{align}
and for univariate noise (one collective coupling to the reservoir) the Hermitian and anti-Hermitian parts are simply the real and imaginary parts.
In either case the anti-Hermitian part of \eqref{eq:Hadamard2} is the second term.
These off-diagonal perturbations perfectly correspond to the results of the previous section upon introducing the appropriate Boltzmann weights as we have.
Moreover we also have the diagonal (and resonant) coefficients determined to be
\begin{align}
& \left. C_{ijk}^{nm} \right|_{\omega_i = \omega_j} = \label{eq:Cinterpolate} \\
&  \mathrm{An}\!\left[ e^{-\beta \omega_k} \frac{d}{d\omega_i} A_{nm}(\omega_{ik}) + e^{-\beta \omega_i} \frac{d}{d\omega_i} A_{mn}(\omega_{ki}) \right] \, , \nonumber
\end{align}
where here the Boltzmann weights are guessed, as these relations have only been derived here for zero-temperature.
Therefore Eq.~\eqref{eq:Cinterpolate} is exact for zero-temperature and our best guess for the positive-temperature coefficients:
it has the correct functional dependence upon the Boltzmann weight and fourth-order master equation coefficients.
At worst this is an interpolation of the zero and high-temperature states.

\section{Discussion}
\label{sec:discussion}

In this work we investigate the equilibrium 
states of open quantum systems from
dynamics / non-equilibrium point of view. 
We show that starting from a product state \eqref{eq:product}
the open system which results from coarse graining the environment will evolve to a late-time steady state. This state can be expressed as the reduced state of the closed system thermal state at the temperature of the environment, i.e. Eq.~\eqref{eq:ER}.
This result is important when the system-environment coupling is not negligible\footnote{Based on the discussion of Fig.~\ref{fig:scale}, we expect our results to be most relevant to small systems.}, or alternatively, when relaxation rates
are not insignificant in relation to the system frequencies. In this case the stationary state of the system \eqref{eq:ER} differs from the
canonical Boltzmann state \eqref{eq:ES}.\footnote{In this paper we have not focused on the nature of this difference. A
quantification in terms of the
\textit{Hamiltonian of mean force} for the special case of an Ohmic environment
is given by Hilt et al. \cite{Hilt2011}. We intend
to address this issue in our future work.}
One might argue that this state is the closest one can get to thermalization in the strong coupling regime.\footnote{Alternatively one could define this state to be \emph{the} thermal state in the strong coupling regime. However this state depends on the specifics of the reservoir and the coupling to the reservoir. Hence it is not specified by the system parameters alone and referring to it as \emph{the} thermal state is, in our opinion, misleading.}
However in this paper we use the term equilibrium state for Eq.~\eqref{eq:ER} and reserve the term thermal state to the standard Boltzmann form \eqref{eq:ES}.

Our proof is exact for the linear model and to second order
in interaction strength for nonlinear models. Moreover, for the exactly solvable
linear case we prove the equivalence of multi-time correlations. The issue of
multi-time correlations in the context of equilibration/thermalization seems to be mostly
ignored in the literature. We argue that multi-time correlations are
important outside the Markovian regime, as was pointed out in \cite{CRV03}. For instance, the relaxation of multi-time correlations
cannot be deduced from the relaxation of the reduced density matrix of the system, neither can the explicit value of the multi-time correlations be derived from the equilibrium state, if the dynamics is non-Markovian. In this respect our analysis of the linear N-QBM model provides insight into equilibration phenomena beyond the density matrix formalism.

A complete proof, which would be non-perturbative for non-linear systems, would
have to be very different than the second-order proof presented here. Our
nonlinear proof, though very general in its application to different systems
and environments, is not robust enough for non-perturbative multi-time
correlations. It is not immediately clear how such a proof could be attempted,
whereas the elegance of the final result makes the possibility of its existence
seem reasonable.

An analogous proof for classical systems should be attempted by coarse graining the symplectomorphic (Hamiltonian) time evolution of the system and
environment in much the same way that quantum master equations result from
coarse graining the unitary time evolution of the system and environment.
Unfortunately the literature on such an analog is not well developed (e.g.,
it would involve higher-order Fokker-Plank equations which might only
perturbatively preserve probability) and
this would be more mathematically challenging than the quantum proof. Note
 that the $\hbar \rightarrow 0$ limit of the quantum results obtained in this
 paper yield the corresponding classical results, as has been argued in
 App.~\ref{sec:CGN}.

An essential ingredient of our proofs is the continuum limit for the
environment. For a finite environment the $t \rightarrow \infty$ limit of the
reduced state does not exist within the formalism presented here
and another ingredient is necessary to ensure relaxation to equilibrium.
Having classical molecular dynamics in mind,
we entertain the possibility that quantum chaos might be one avenue to explore.

On the other hand we can consider a large but finite environment. It can be argued that for any relevant times $t > 0$ the effect of an infinite reservoir can be approximated arbitrarily closely by a large but finite reservoir. Then equilibration is observed for the time-interval between the relaxation time and the recurrence time. Note that this interval is huge for a large environment, since the recurrence time grows very rapidly with the number of degrees of freedom. As a result the system stays close to its equilibrium most of the time. 
This interpretation helps us touch base with the results of \cite{Linden09,Reimann10,Short12} where relaxation in finite systems is proven for time averaged quantities.

\subsection{Comparison with recent literature}
To put this work in developmental context,
here we compare more specifically our results to that of Linden et al. \cite{Linden09}, Reimann \cite{Reimann10}, and Short and Ferrelly \cite{Short12}\footnote{See Sec.~\ref{sec:terminology} for the clarification of the different use of the term \emph{equilibration} in the literature and here.}. All these works have in common with us the set-up of a small system coupled to a large environment and relaxation is achieved dynamically via time-evolution. A major difference is the choice of initial conditions: they allow for any initial state, which is spread over sufficiently many energies, whereas we restrict our environment to be in a thermal state. In turn we can derive the form of the equilibrium state explicitly.

Unlike what is done here these authors all make the assumption of non-degenerate energy gaps (this assumption is relaxed to a certain degree in \cite{Short12}) and assume finite dimensional Hilbert spaces. The linear model we solved exactly here has infinitely degenerate energy gaps and we considered a reservoir consisting of an infinite number of degrees of freedom. Ref. \cite{Linden09} 
considers only pure states for the closed system (in the sprit of \cite{Popescu06,Goldstein06}). Finally they all define relaxation in terms of time averaged quantities, i.e. systems behave as if they are in their steady state most of the time. Ref. \cite{Short12} also provides upper limit for the relaxation time.

The proofs of \cite{Linden09,Short11,Short12,Reimann10} rely on the much greater dimensionality of the Hilbert space of the environment compared to that of the system.
The system + environment state is propagated as a whole using unitary dynamics. The fact that the environment is large is utilized in the tracing out of the environment at the end of time evolution. In this derivation the effect of the environment on the system dynamics is not so easily accessible.

In our proof, the fact that the environment consists of a large number of degrees of freedom manifests itself in the form of its decaying correlations. These correlations in turn determine the non-unitary aspects of the open system dynamics. We use this non-unitary open system dynamics to evolve the reduced state of the system to its equilibrium state. In particular we do not refer to the state of the closed system explicitly\footnote{Except for Sec.~\ref{sec:coarse}, where we do look at the individual environmental modes just to make the point that the closed system (S + E) does not equilibrate.}. Our derivation is more in the idioms of open quantum systems paradigm, where the influence of the environment on the system dynamics can be continuously monitored and explicitly expressed  (e.g., consistent backreaction from the environment is fully embodied in the influence functional \cite{Feynman63}).

Relaxation is demonstrated in  \cite{Linden09,Reimann10,Short12} for very general Hamiltonians, including strong coupling between the system and the environment. In their derivation the strong coupling regime does not present any extra difficulty. In the open system approach we adopted in this paper strong coupling is difficult to handle. On the other hand, as a benefit of our method we can describe the nature of the equilibrium state, i.e. Eq.~\eqref{eq:ER}, besides proving its existence and uniqueness.

\appendix

\section{The triviality of classical, Gaussian noise}\label{sec:CGN}

While Ref.~\cite{Geva00} gives many cases in quantum mechanics in which the effect of system-environment coupling on the equilibrium state may be overlooked,
here we would like to motivate the fact that this point is
often overlooked in the classical regime as well,
perhaps due to the ubiquitous employment of Gaussian noise.
Let us consider the Hamiltonian of a system coupled linearly, via the system operator $\mathbf{L}$,
to an environment of harmonic oscillators, indexed by $k$, which mock our Gaussian noise \cite{Feynman63,CaldeiraLeggett81}.
\begin{align}
\mathbf{H}_\mathrm{C} &= \mathbf{H}_\mathrm{S} + \sum_k \left[ \frac{\mathbf{p}_k^2}{2m_k} + \frac{m_k \omega_k^2}{2} \mathbf{x}_k^2 \right] + \mathbf{L} \sum_k g_k \mathbf{x}_k + \mathbf{H}_\mathrm{R} \, , \nonumber \\
&= \mathbf{H}_\mathrm{S} + \sum_k \left[ \frac{\mathbf{p}_k^2}{2m_k} +
\frac{m_k \omega_k^2}{2} \left( \mathbf{x}_k - \frac{g_k \mathbf{L}}{m_k
\omega_k^2} \right)^2 \right] \, ,
\end{align}
where the linear interaction is included in the square of the environment
potential as a means of ``renormalization".
Otherwise, the influence of the environment effectively introduces a negative $\mathbf{L}^2$ term proportional to the cutoff into the system Hamiltonian
when considering the open-system dynamics.

Tracing over the environmental degrees of freedom is equivalent to integrating over the environmental dimensions in phase space,
\begin{align}
\mathrm{Tr}_\mathrm{E} \!\left[ \cdots \right] &= \prod_k \int \!\! \rd x_k
\int \!\! \rd p_k \cdots \, ,
\end{align}
where classically-speaking, $x_k$ and $p_k$ are independent, commuting variables.
Therefore, in the classical and Gaussian model, relations \eqref{eq:ES} and \eqref{eq:ER} are equivalent
as tracing over the environmental degrees of freedom constitutes a trivial
Gaussian integral in phase space. The classical result can also be reached as
the $\hbar \rightarrow 0$ limit of the quantum result. This limit is most
straight-forward when applied to the Wigner function \cite{Hillery84} defined
as:
\begin{align}
W\!\left( x,p\right) = \frac{1}{2\pi \hbar} \int \!\! \rd u \,
e^{\frac{\iota}{\hbar}p u} \rho(x \!-\! u/2,x \!+\! u/2) \, .
\end{align}
The description in terms of the Wigner function is equivalent to the density
matrix approach. Hence the Wigner function contains complete information about
the quantum system. As a result the Wigner function \textit{should not} be
treated as a phase space distribution, since it can assume negative values.
However the $\hbar\rightarrow 0$ limit of the thermal state Wigner function is
well-defined and gives the classical Boltzmann distribution function:
\begin{align}
 \lim_{\hbar\rightarrow 0}W_\beta(x,p)= \frac{e^{-\beta H(x,p)}}{Z(\beta)} \, .
\end{align}

For classical open systems it is well known that if the system + environment is
in a thermal state of the full Hamiltonian, which includes the system-environment coupling, then the reduced distribution of the system is in general not the
thermal distribution of the system Hamiltonian alone. The term
\textit{potential of mean force} is used in chemical-physics literature for the
quantity that replaces the Hamiltonian in the familiar Boltzmann distribution
\cite{Kirkwood35}.
The linear reservoir is a special case where the potential of mean
force coincides with the system Hamiltonian. The potential of mean force is
defined by\footnote{In most treatments $\bH_\rR$ is absent. In that case even
for linear reservoir $\bH^*$ differs from $\bH_\rS$ by a frequency
``renormalization".}:
\begin{align}
& \mathbf{H}^*(\mathbf{X},\mathbf{P}) \equiv - \frac{1}{\beta} \log
\frac{\prod_k \int \!\! \rd x_k \int \!\! \rd p_k \, e^{-\beta \,
\bH_{\rC}(\bX,\bP;\bx,\bp) }}
{\prod_{k'} \int \!\! \rd x_{k'} \int \!\! \rd p_{k'} \, e^{-\beta \,
\bH_\rE(\bx,\bp)}} \, .
\end{align}

To the best of our knowledge,
the asymptotic time evolution of a general classical open system, with a
nonlinear environment initially in its thermal state, is not
known.
We conjecture that the reduced system is asymptotically described by $e^{-\beta \bH^*}$ as described in the previous paragraph,
and as would follow from \eqref{eq:ER}.
In this paper we provide a proof of the analogous statement for quantum systems to second order in interaction strength.
Obviously, our second-order proof extends to classical systems which can arise in the limit $\hbar \to 0$.
For linear systems we have an exact proof,
and unlike its classical counterpart, the quantum linear case is highly nontrivial.

\section{Theorems on matrix derivatives}
\label{app:matrixtheorems}

\textit{Notation and Remarks}: A letter in bold like $\bA$ indicates a matrix.
Referring to an element of the matrix we use subscripts: $\bA_{\ra \rb}$. The
inverse of the matrix is indicated by
$\bA^{-1}$. An element of the inverse matrix is written as $(\bA^{-1})_{\ra
\rb}$
to avoid confusion with $1/\bA_{\ra \rb}$. Transpose of the matrix is denoted
by $\bA^{\mathrm{T}}$. $\tr$ without a subscript indicates ordinary matrix
trace. $\tr_\rC$ indicates quantum mechanical trace over the closed system
Hilbert space. A systematic study of matrix derivatives including some of the
theorems below is given by \cite{Dwyer67}.

Before proceeding to the derivations we clarify a mathematical subtlety. The
theorems derived in this appendix will mostly be applied to symmetric matrices
for which $\bA_{\ra \rb}=\bA_{\rb\ra}$. When taking the derivative of such a
matrix with respect to one of its elements one can adopt two
different conventions.
If the derivative is taken under the constraint that only symmetric variations
of the matrix is allowed the result is:
\begin{align}
 \frac{\del \bA_{\ra \rb}}{\del
\bA_{\rc\rd}}=\delta_{\ra
\rc}\delta_{\rb \rd}+\delta_{\ra\rd}\delta_{\rb\rc}(1-\delta_{\ra \rb}) \, .
\end{align}
On the other hand if independent variations of all matrix elements are allowed
the second term in the above equation is absent. In the following theorems we
adopt the second convention.

\begin{theorem}
\label{thm:difftrick}
Consider a system in a thermal state
at inverse temperature $\beta$ described by a
Hamiltonian with parametric dependence on a set of variables $\{\lambda_n\}$.
Then the expectation value of the derivative of the Hamiltonian with respect to these
parameters can be calculated from the partition function by:
\begin{align}
\label{eq:difftrick}
\left \langle \frac{\partial \bH }{\partial \lambda_n} \right \rangle_\rC
\equiv \tr_\rC \! \left[ \frac{\partial \bH}{\partial \lambda_n}
\frac{e^{-\beta \bH}}{Z} \right]
=-\frac{1}{\beta}\frac{\partial}{\partial \lambda_n}\ln(Z) \, .
\end{align}
\end{theorem}
\begin{proof}
In this proof we will make use of the following operator
identity valid for an arbitrary operator $\boldsymbol{\mathcal{O}}$:
\begin{align}
\frac{\partial}{\partial \lambda_n} e^{\boldsymbol{\mathcal{O}}} = \int_0^1
\rd u \, e^{u\boldsymbol{\mathcal{O}}} \, \frac{\partial
\boldsymbol{\mathcal{O}}}{\partial\lambda_n} \,
e^{(1-u)\boldsymbol{\mathcal{O}}} \, .
\end{align}
Using this formula we can write the RHS of Eq.~\eqref{eq:difftrick} as:
\begin{align}
-\frac{1}{\beta}\frac{\partial}{\partial \lambda_n}\ln(Z)
&=-\frac{1}{\beta Z}\tr_\rC \! \left[\frac{\partial }{\partial \lambda_n}e^{-\beta \bH}\right] , \\
&=\frac{1}{\beta Z}\tr_\rC \! \left[-\int_0^1 \rd u \, e^{-u \beta\bH} \,
\frac{\partial \beta \bH}{\partial\lambda_n} \, e^{-(1-u)\beta\bH}\right] .
\end{align}
We use the cyclic property of trace to get:
\begin{align}
 -\frac{1}{\beta}\frac{\partial}{\partial \lambda_n}\ln(Z)
&= \frac{1}{ Z}\tr_\rC \! \left[\int_0^1 \rd u \, \frac{\partial \bH}{\partial
\lambda_n} \,  e^{-\beta\bH} \right] , \\
&= \frac{1}{Z}\tr_\rC \! \left[ \frac{\partial \bH}{\partial \lambda_n} \, e^{-\beta\bH}\right] , \\
& =\left \langle \frac{\partial \bH}{\partial \lambda_n} \right \rangle_\rC .
\end{align}
\end{proof}

\begin{theorem}
\label{thm:trlog}
For a matrix $\bA$
\begin{align}
 \tr \log \mathbf{A}=\log \det\mathbf{A} \, .
\end{align}
\end{theorem}
\begin{proof}
Trace operation is basis-independent. In the basis in which
$\mathbf{A}$ is diagonal $\log\mathbf{A}$ is also a diagonal matrix
with entries $\log a_n$ where $a_n$ are the eigenvalues of $\mathbf{A}$.
Taking the trace gives:
\begin{align}
 \tr \log \mathbf{A}=\sum_n \log a_n=\log \left(\prod_n
a_n\right) \, .
\end{align}
The last expression is recognized to be $\log \det\mathbf{A}$ since the
product of eigenvalues equals the determinant.
\end{proof}

\begin{theorem}
\label{thm:TrLogPerm}
For an arbitrary number of matrices
$\bA_k$ indexed by $k$, the following is true:
\begin{align}
\tr \log\!\left( \prod_k \bA_k \right) = \sum_k \tr \log\!\left(\mathbf{A}_k\right) .
\end{align}
\end{theorem}
\begin{proof}
To show this equality we make use of
Theorem~\ref{thm:trlog}, the well known fact that
the determinant of the product of matrices equals the product of the
determinants and properties of ordinary logarithms:
\begin{align}
\tr \log\!\left( \prod_k \bA_k \right)
&= \log \det \!\left( \prod_k \bA_k \right) , \\
&= \log \!\left( \prod_k \det \bA_k \right) ,  \\
&= \sum_k \log \det \bA_k \, , \\
&= \sum_k \tr \log\!\left(\mathbf{A}_k\right)  .
\end{align}
\end{proof}
A corollary of this theorem is the fact that $\tr \log$ is invariant under
any permutation of its arguments.

\begin{theorem}
\label{thm:dertrlog}
Consider a matrix $\mathbf{A}$ and a
parameter $\lambda$.
Then:
\begin{align}
\label{eq:dertrlog1}
 \frac{\partial}{\partial \lambda}\tr\log \mathbf{A} =
\tr\left[\mathbf{A}^{-1} \frac{\partial \mathbf{A}}{\partial \lambda}
\right] \, .
\end{align}
In particular:
\begin{align}
\label{eq:dertrlog2}
 \frac{\partial}{\partial \mathbf{A}}\tr\log \mathbf{A} =
(\mathbf{A}^{-1})^\mathrm{T} \, .
\end{align}
where $\frac{\partial}{\partial \mathbf{A}}$ is defined as the matrix obtained
by differentiating with respect to the entries of matrix $\mathbf{A}$.
\end{theorem}
\begin{proof}
Let $\bA\equiv 1+\bB$ and use
\begin{equation}
\log(1+\bB)=\bB-\bB^2/2 +\bB^3/3+\cdots \, ,
\end{equation}
to write the LHS of
Eq.~\eqref{eq:dertrlog1} as:
\begin{widetext}
\begin{align}
\frac{\partial}{\partial \lambda}\tr\left[ \bB-\frac{\bB^2}{2}+\frac{\bB^3}{3} + \cdots  \right] &=
\tr\left[ \frac{\partial \bB}{\partial \lambda}-\frac{1}{2}(\frac{\partial \bB}{\partial \lambda} \bB+\bB \frac{\partial \bB}{\partial \lambda})+\frac{1}{3}(\frac{\partial \bB}{\partial \lambda}\bB^2+\bB\frac{\partial \bB}{\partial \lambda}\bB+\bB^2\frac{\partial \bB}{\partial \lambda}) + \cdots  \right] .
\end{align}
\end{widetext}
Using the cyclic property of trace we obtain:
\begin{align}
 \tr\left[ \frac{\partial \bB}{\partial
\lambda}\left(1-\bB+\bB^2-\bB^3\cdots \right)  \right] \, .
\end{align}
Note that $\partial\bB /\partial\lambda=\partial\mathbf{A}/\partial\lambda$ and
\begin{equation}
1-\bB+\bB^2-\bB^3 + \cdots = (1+\bB)^{-1} = \bA^{-1} \, ,
\end{equation}
which proves Eq.~\eqref{eq:dertrlog1}.
To prove Eq.~\eqref{eq:dertrlog2} let $\lambda\equiv \mathbf{A}_{\ra \rb}$.
\begin{align}
 \frac{\partial}{\partial \mathbf{A}_{\ra \rb}}\tr\log \mathbf{A}
&=\tr\left[ \mathbf{A}^{-1} \frac{\partial \mathbf{A}}{\partial \mathbf{A}_{\ra \rb}} \right] , \\
&=\sum_{\rc\rd} \left(\mathbf{A}^{-1}\right)_{\rc\rd} \frac{\partial \mathbf{A}_{\rd\rc}}{\partial \mathbf{A}_{\ra \rb}} \, , \\
&=\left(\mathbf{A}^{-1}\right)_{\rb\ra} \, .
\end{align}
\end{proof}

\begin{theorem}
\label{thm:diffinv}
Let $\mathbf{A}$ be an invertible
matrix and $\lambda$ a parameter. Then:
\begin{align}
\label{eq:first}
 \frac{\partial \mathbf{A}^{-1}}{\partial
\lambda}=-\bA^{-1}\frac{\del \bA}{\del \lambda}\bA^{-1} \, .
\end{align}
In particular:
\begin{align}
\label{eq:second}
 \frac{\partial (\mathbf{A}^{-1})_{\ra \rb}}{\partial
\mathbf{A}_{\rlm \rn}}=-(\mathbf{A}^{-1})_{\ra\rlm}(\mathbf{A}^{-1})_{\rn\rb}
\, .
\end{align}
\end{theorem}
\begin{proof}
We write $\mathbf{A}^{-1}=\mathbf{A}^{-1}\mathbf{A} \, \mathbf{A}^{-1}$, and differentiate both sides with respect to $\lambda$.
Looking at an element of this matrix equation we have:
\begin{widetext}
\begin{align}
\frac{\partial (\mathbf{A}^{-1})_{\ra \rb}}{\partial \lambda}
&= \sum_{\rc\rd}\left(
\frac{\partial
  (\mathbf{A}^{-1})_{\ra\rc}}{\partial
  \lambda}\mathbf{A}_{\rc\rd}(\mathbf{A}^{-1})_{\rd\rb} +
  (\mathbf{A}^{-1})_{\ra\rc}\frac{\partial \mathbf{A}_{\rc\rd}}{\partial
  \lambda}(\mathbf{A}^{-1})_{\rd\rb} +
  (\mathbf{A}^{-1})_{\ra\rc}\mathbf{A}_{\rc\rd}\frac{\partial
  (\mathbf{A}^{-1})_{\rd\rb}}{\partial
  \lambda}   \right) ,\\
&= \frac{\partial (\mathbf{A}^{-1})_{\ra \rb}}{\partial
  \lambda} + \left(\mathbf{A}^{-1} \frac{\del
\bA}{\del \lambda}\mathbf{A}^{-1}\right)_{\!ab} +
  \frac{\partial (\mathbf{A}^{-1})_{\ra \rb}}{\partial \lambda} \, .
\end{align}
\end{widetext}
This proves Eq.~\eqref{eq:first}.
For the proof of Eq.~\eqref{eq:second} we set $\lambda=\bA_{\rlm\rn}$.
\end{proof}

A corollary of this theorem is the following identity valid for independent
matrices $\bA_k$:
\begin{align}
 \frac{\partial}{\partial \mathbf{A}_1}\tr\left[ \mathbf{A}_2 \,
\mathbf{A}_1^{-1}
\mathbf{A}_3 \right]=-\left( \mathbf{A}_1^{-1} \mathbf{A}_3 \, \mathbf{A}_2 \,
\mathbf{A}_1^{-1} \right)^{\rT} .
\end{align}

\section{N-QBM Partition Function}
\label{app:partitionfunction}

In this section we calculate the partition function of the N-QBM model. Our
treatment
mimics and generalizes that of Weiss \cite{Weiss93}, which treats one system
oscillator only
and does not allow for interactions among reservoir oscillators and non-diagonal
mass matrix.\footnote{Since a set of non-interacting oscillators can
represent the most general Gaussian thermal reservoir, considering a non-diagonal
mass matrix may appear superfluous. However we need the nondiagonal
elements to generate the correlation function of two different reservoir momenta by
partial differentiation of the partition function.}
The partition function has an imaginary-time path integral representation given by:
\begin{align}
 Z_\rC&=\oint \mathcal{D}\mathbf{x} \, \mathcal{D}\mathbf{X}
\exp\left( - S^{(\rE)}[\mathbf{x},\mathbf{X}]/\hbar \right) \, ,\\
S^{(\rE)} &= \int_0^{\hbar \beta}\rd  \tau \, L^{(\rE)}(\tau) \, ,\\
L^{(\rE)}(\tau)&=\frac{1}{2} \left( \dot{\mathbf{X}}^\mathrm{T} \mathbf{M}
\,
\dot{\mathbf{X}} + \mathbf{X}^\mathrm{T} \mathbf{C} \, \mathbf{X} \right) \\
\nonumber
& + \frac{1}{2} \left( \dot{\mathbf{x}}^\mathrm{T} \mathbf{m} \,
\dot{\mathbf{x}} + \left[ \mathbf{x} - \mathbf{c}^{-1} \mathbf{g} \, \mathbf{X}
\right]^\mathrm{T} \!\! \mathbf{c} \left[ \mathbf{x} - \mathbf{c}^{-1}
\mathbf{g} \, \mathbf{X} \right] \right) \, ,
\end{align}
where $S^{(\rE)}$ is the Euclidean action, $\tau$ the imaginary time and the
path
integral is over
 all periodic trajectories in the interval $[0,\hbar \beta]$.
This path integral is Gaussian and can be evaluated exactly. It is
convenient to represent the integration paths via their Fourier
 series, which takes care of the condition on periodicity.
 
\begin{align}
 \mathbf{x}(\tau)&= \sum_{\rr=-\infty}^{\infty} \mathbf{x}_\rr \, e^{\imath \nur
\tau} \, ,\\
 \mathbf{X}(\tau)&=\sum_{\rr=-\infty}^{\infty} \mathbf{X}_\rr \, e^{\imath \nur
\tau} \, ,
\end{align}
 where $\mathbf{x}_{-\rr}=\mathbf{x}_\rr^\dagger$,
$\mathbf{X}_{-\rr}=\mathbf{X}_\rr^\dagger$ (dagger stands for Hermitian
conjugation) since $\mathbf{x}(\tau)$ and
$\mathbf{X}(\tau)$ are real and $\nu_\rr \equiv 2\pi \rr/\hbar \beta$ are the
bosonic Matsubara frequencies.
 Written in terms of the Fourier coefficients the Euclidean action becomes:
 \\
 \\
 \\
\begin{widetext}
\begin{align}
S^{(\rE)}
&= \frac{\hbar \beta}{2} \sum_{\rr=-\infty}^{\infty} \left( \mathbf{X}_\rr^{\dagger} \, (\nur^2 \, \mathbf{M} +\mathbf{C}) \, \mathbf{X}_\rr \right)
+ \frac{\hbar \beta}{2} \sum_{\rr=-\infty}^{\infty}\left( \mathbf{x}_\rr^\dagger \, \nur^2 \, \mathbf{m} \, \mathbf{x}_\rr + \left[
\mathbf{x}_\rr - \mathbf{c}^{-1} \mathbf{g} \, \mathbf{X}_\rr \right]^\dagger
\!\! \mathbf{c} \left[ \mathbf{x}_\rr - \mathbf{c}^{-1} \mathbf{g} \,
\mathbf{X}_\rr \right] \right) \, .
\end{align}
\end{widetext}
Next we decompose $\mathbf{x}_\rr=\bar{\mathbf{x}}_\rr+\mathbf{y}_\rr$  where
\begin{equation}
\bar{\mathbf{x}}_\rr=(\nur^2 \,\mathbf{m}+\mathbf{c})^{-1}\mathbf{g} \, \mathbf{X}_\rr \, ,
\end{equation}
is  chosen such that $S^{(\rE)}$ does not have a term linear in
 $\mathbf{y}_\rr$. The action can be written as:
\begin{align}
S^{(\rE)}
&= S^{(\rE)}_{\mathrm{reservoir}}[\mathbf{y}]
+S^{(\rE)}_{\mathrm{system}}[\mathbf{X}] \, , \\
&=\frac{\hbar \beta}{2} \sum_{\rr=-\infty}^\infty \left(\mathbf{y}_\rr^{\dagger} \, (\mathbf{\nur^2 \, \mathbf{m}+c})\, \mathbf{y}_\rr  \right)\\
\nonumber &+\frac{\hbar \beta}{2}  \sum_{\rr=-\infty}^\infty \left(
\mathbf{X}_\rr^{\dagger} \, (\nur^2 \, \mathbf{M} +\mathbf{C}+2 \, \nur
\boldsymbol{\hat{\gamma}}(\nur)) \, \mathbf{X}_\rr \right) \, ,
\end{align}
where the damping kernel is given by
\begin{align}
\boldsymbol{\hat{\gamma}}(z)
=\frac{1}{2} \, \mathbf{g}^{\rT}\mathbf{m}^{-\frac{1}{2
}}\boldsymbol{\omega}^{-1}\frac{z}{\boldsymbol{\omega}^2+z^2}\boldsymbol{\omega
}^{ -1}
\mathbf { m}^{-\frac{1}{2}}\mathbf{g} \, ,
\end{align}
which is the Laplace transform of Eq.~\eqref{eq:dampkern}. The
partition function
of the closed system is given by:
\begin{align}
Z_\rC &=\mathcal{N} \int \prod_{\rr=-\infty}^{\infty}  \rd \mathbf{X}_\rr \exp\left(-S^{(\rE)}_{\mathrm{system}}[\mathbf{X}]/\hbar\right) \nonumber \\
& \times \int \prod_{\rr=-\infty}^{\infty} \rd \mathbf{y}_\rr  \exp\left(-S^{(\rE)}_{\mathrm{reservoir}}[\mathbf{y}]/\hbar\right) \, . \label{eq:Zint}
\end{align}
The normalization factor $\mathcal{N}$
 is yet unspecified because it is not easy to
determine
the measure of the path integral. $\mathcal{N}$ will be determined
indirectly at the final stage of this calculation
by considering the limiting case of no system-environment coupling.

The integrals in Eq.~\eqref{eq:Zint} are all Gaussian. Ignoring the
normalization for now the integration gives:
\begin{widetext}
\begin{align}
 Z_\rC
&\propto  \prod_{\rr=-\infty}^{\infty} \frac{1}{\sqrt{\det\left[\nur^2 \, \mathbf{m}+\mathbf{c} \right]}}
\frac{1}{\sqrt{\det\left[\nur^2 \, \mathbf{M}+\mathbf{C}+2 \, \nur \hat{\boldsymbol{\gamma}}(\nur) \right]}} \, , \\
 & \propto \frac{1}{\sqrt{\det[\mathbf{c}]}}
\frac{1}{\sqrt{\det[\mathbf{C}]}} \prod_{\rr=1}^{\infty}
\frac{1}{\det\left[\nur^2 \, \mathbf{m}+\mathbf{c} \right]}
\frac{1}{\det\left[\nur^2 \, \mathbf{M}+\mathbf{C}+2\, \nur
\hat{\boldsymbol{\gamma}}(\nur) \right]}\, .
\end{align}
\end{widetext}
In the second line we used the fact that the elements of the product
corresponding to positive and negative values of $\rr$ are identical to
restrict
the product to positive $\rr$ and pulled out the $\rr=0$ entry. To determine
the
normalization let us recall the partition function for a simple
harmonic oscillator:
\begin{align}
 Z_{\mathrm{1HO}}=\frac{1}{2 \sinh(\beta \hbar \omega/2)}=\frac{1}{\beta \hbar
\omega}
\prod_{\rr=1}^{\infty}\frac{\nur^2}{\omega^2+\nur^2} \, .
\end{align}
This naturally generalizes to N harmonic oscillators by:
\begin{align}
  Z_{\mathrm{NHO}}=\frac{1}{\det[ 2 \sinh(\beta \hbar
\boldsymbol{\omega}/2)]}=\frac{1}{\beta
\hbar \det[\boldsymbol{\omega}]}
\prod_{\rr=1}^{\infty}\frac{\nur^2}{\det[\boldsymbol{\omega}^2+\nur^2]} \, .
\end{align}
In the limit of no coupling we demand that the partition function be a product
of two partition functions of this form. This condition fixes the
normalization and the final answer is:
\begin{align}
 Z_\rC &= Z_\rE \times \det \!\left( \frac{1}{\beta\hbar\boldsymbol{\Omega}}\right) \\
&\times \prod_{\rr=1}^\infty \det\!\left( \frac{\nur^2}{\boldsymbol{\Omega}^2+\nur^2+2 \mathbf{M}^{-\frac{1}{2}}\nur\hat{\boldsymbol{\gamma}}(\nur)\mathbf{M}^{ -\frac { 1 } { 2 } } } \right) \, , \nonumber
\end{align}
where $Z_\rE=\tr \left[ \exp(-\beta \mathbf{H}_\mathrm{E}) \right]$
is the partition function
of reservoir oscillators without coupling to the system. Using the
definition\,\eqref{eq:Ghat}
the partition function can also be written as:
\begin{align}
\label{eq:PF}
 Z_\rC = Z_\rE \times \det\left(
\frac{1}{\hbar\beta\boldsymbol{\Omega}}\right)
\prod_{\rr=1}^\infty \det\left(\mathbf{M} \nur^2
\hat{\mathbf{G}}(\nur)\right) \, .
\end{align}

\section{Derivation of Eqs.~(\ref{eq:<XX>}-\ref{eq:<xx>})}
\label{app:singletime}

In this appendix we derive some of the results presented in Sec.~\ref{sec:singletimecorrelations}.
Angular bracket with the subscript C denotes expectation values in the closed system
thermal state. Expectation values in the uncorrelated state are denoted by
attaching the subscript E to the bracket. Note that the
damping kernel depends on
the environmental
variables and the coupling constants alone. There is no dependence on system
variables. Using Eq.~\eqref{eq:XXform} we calculate the single-time system
position-position correlation as:
\begin{widetext}
\begin{align}
\lbr \left(\mathbf{X} \, \mathbf{X}^{\rT}\right)_{\rA\rB}\rbr
&=\frac{1}{\beta}\frac{\partial}{\partial \mathbf{C}_{\rA\rB}}\tr \log \mathbf{C}
+\frac{2}{\beta}\frac{\partial}{\partial \mathbf{C}_{\rA \rB}}\sum_{r=1}^\infty \tr \log \left[ \mathbf{M}^{-1} \hat{\mathbf{G}}(\nur)^{-1}  \right] , \\
 &= \frac{1}{\beta}\left( \mathbf{C}^{-1}\right)_{\rA \rB}+\frac{2}{\beta}\sum_{r=1}^\infty\tr \left[ \left( \mathbf{M}^{-1} \hat{\mathbf{G}}(\nur)^{-1} \right)^{-1} \mathbf{M}^{-1} \frac{\partial \hat{\mathbf{G}}(\nur)^{-1}}{\partial \mathbf{C}_{\rA \rB}} \right] , \\
&=\frac{1}{\beta} \hat{\mathbf{G}}(\boldsymbol{\nu}_0)_{\rA \rB}+\frac{2}{\beta} \sum_{r=1}^\infty\hat{\mathbf{G}}(\nur)_{\rA \rB} \, , \\
\lbr \mathbf{X} \, \mathbf{X}^{\rT}\rbr
&=\frac{1}{\beta} \hat{\mathbf{G}}(\boldsymbol{\nu}_0)+\frac{2}{\beta} \sum_{r=1}^\infty\hat{\mathbf{G}}(\nur) \, ,
\end{align}
\end{widetext}
where we used the fact that $\mathbf{C}$ and $\hat{\mathbf{G}}(\nur)$ are symmetric matrices and $\bGhat(\bnu_0)=\bGhat(0)=\bC^{-1}$.
The system momentum-momentum correlations can be calculated in a similar way using Eq.~\eqref{eq:PPform}.
\begin{widetext}
\begin{align}
\lbr \left( \mathbf{P} \, \mathbf{P}^{\rT} \right)_{\rA \rB}\rbr
&=\frac{1}{\beta}\frac{\partial}{\partial
(\mathbf{M}^{-1})_{\rA \rB}}\tr \log (\mathbf{M}^{-1})+\frac{2}{\beta}\frac{\partial}{ \partial
(\mathbf{M}^{-1})_{\rA \rB}}\sum_{r=1}^\infty \left(\tr \log \left[
\mathbf{M}^{-1}\right]+\tr\log \left[ \hat{\mathbf{G}}(\nur)^{-1}
\right]\right) \, , \\
 &= \frac{\bM_{\rA\rB}}{\beta}+\frac{2}{\beta}\sum_{\rr=1}^{\infty}\left(
\bM_{\rA\rB} + \tr \left[ \bGhatnur \frac{\del \bGhatnur^{-1}}{\del
(\mathbf{M}^{-1})_{\rA \rB}} \right] \right) \, , \\
\lbr \mathbf{P} \, \mathbf{P}^{\rT} \rbr&=
\frac{\mathbf{M}}{\beta}+\frac{2}{\beta}\sum_{r=1}^\infty \left(
\mathbf{M}-\mathbf{M} \, \nur^2\hat{\mathbf{G}}(\nur) \, \mathbf{M} \right) \, .
\end{align}
\end{widetext}
We used Theorem~\ref{thm:TrLogPerm} in the first line. In the second
line we used Theorem~\ref{thm:dertrlog} for all terms and
Theorem~\ref{thm:diffinv} for the last term with
$\mathbf{A}_1,\mathbf{A}_2,\mathbf{A}_3 \
\rightarrow \mathbf{M}^{-1}, \nur^2\hat{\mathbf{G}}(\nur), \mathbf{1}$.

For the system-environment position correlations note that only the damping
kernel depends on the interaction matrix:
\begin{align}
\nonumber
\frac{1}{\beta}\frac{\partial \log Z_C}{\partial (\mathbf{g}^{\rT})_{\rA\ra}}
&=-\frac{1}{\beta}\sum_{\rr=1}^{\infty} \frac{\partial}{\partial (\mathbf{g}^{\rT})_{\rA\ra}} \tr \log [\hat{\mathbf{G}}(\nur)^{-1}] \, , \\
\label{eq:dellogdelg}
&= -\frac{2}{\beta}\sum_{\rr=1}^\infty
\sum_{\rB\rC}
\hat{\mathbf{G}}(\nur)_{\rB\rC}
\nur \frac{\partial \hat{\boldsymbol{\gamma}}(\nur)_{\rC\rB}}{\partial
(\mathbf{g}^{\rT})_{\rA\ra}} \, .
\end{align}
The partial derivative of the damping kernel can be calculated explicitly. For
this differentiation it is useful to rewrite $\bgammahatnur$ as:
\begin{align}
 2 \nur \bgammahatnur= \nur^2 \bgt \bctomo \left( \bmtomo+\nur^2
\bctomo\right)^{-1} \bctomo \bg
\end{align}
\\
For brevity of notation we define $\basr$ such that
$\bgammahatnur=\frac{1}{2}\bgt \nur \basr \, \bg$.
\begin{widetext}
\begin{align}
2 \nur \frac{\partial \hat{\boldsymbol{\gamma}}(\nur)_{\rC\rB}}{\partial (\mathbf{g}^{\rT})_{\rA\ra}} &=
\frac{\partial}{\partial \mathbf{g}_{\ra\rA}}\sum_{\re\rf}  \left( \mathbf{g}^{\rT} \right)_{\rC\re}\left(\nur^2 \basr\right)_{\re\rf}\mathbf{g}_{\rf\rB} \, , \\
&=\sum_{\re\rf}\bigg\{  \delta_{\re\ra} \delta_{\rC\rA}\left(\nur^2 \basr\right)_{\re\rf} \bg_{\rf\rB} +  \left( \mathbf{g}^{\rT} \right)_{\rC\re} \left(\nur^2\basr\right)_{\re\rf}\delta_{\rf\ra}\delta_{\rB\rA} \bigg\} , \\
&=  \left(\nur^2 \basr \, \bg \right)_{\ra\rB}\delta_{\rC\rA} +   \left(\bgt \nur^2 \basr \right)_{\rC\ra}\delta_{\rB\rA} \, .
\end{align}
\end{widetext}
Plugging this result in Eq.~\eqref{eq:dellogdelg} we get:
\begin{align}
\frac{1}{\beta}\frac{\partial
\log Z_C}{\partial
\mathbf{g}^{\rT}}=-\frac{2}{\beta}\rsum \nur \bGhatnur \, \bgammahatnur \,
\bg^{-1} \, .
\end{align}
Using this result in Eq.~\eqref{eq:Xxform} we get Eq.~\eqref{eq:<xX>}.

To derive Eq.~\eqref{eq:<pp>} we start from Eq.~\eqref{eq:ppform}:
\begin{align}
\label{eq:ppf}
 -\frac{2}{\beta} \frac{\del \log Z_C}{\del \bmtomo}
&= \left \langle \bp \, \bpt \right \rangle_{\rE} \\
& + \frac{2}{\beta} \rsum \frac{\del}{\del \bmtomo}\tr \log \left[ \bM^{-1} \bGhat^{-1}(\nur) \right] . \nonumber
\end{align}
Using Theorem\,\ref{thm:dertrlog} we get:
\begin{align}
\label{eq:tr[GdelG]}
 \frac{\del}{\del (\bmtomo)_{\ra \rb}}\tr \log \left[ \bM^{-1}
\bGhat^{-1}(\nur)
\right]=\tr\left[ \bGhatnur \frac{\del \bGhat^{-1}(\nur)}{\del
(\bmtomo)_{\ra \rb}} \right] \, ,
\end{align}
where
\begin{align}
 \frac{\del \bGhat^{-1}(\nur)}{\del (\bmtomo)_{\ra \rb}}=2 \nur \frac{\del
\bgammahatnur}{\del (\bmtomo)_{\ra \rb}} \, .
\end{align}

Next use Theorem~\ref{thm:diffinv} and plug the result back into Eq.~\eqref{eq:tr[GdelG]}:
\begin{widetext}
\begin{align}
2 \nur \frac{\del
\bgammahatnur}{\del (\bmtomo)_{\ra \rb}}&=
-\nur^2 \bgt \basr^{-1}  \frac{\del \basr}{\del (\blm^{-1})_{\ra\rb}}
\basr^{-1} \bg \, , \\
&=-\nur^2 \bgt \bctomo (\bmtomo+\nur^2
\bctomo)^{-1} \frac{\del(\bmtomo +\nur^2 \bctomo)}{\del
(\bmtomo)_{\ra \rb}}(\bmtomo+\nur^2
\bctomo)^{-1} \bctomo \bg \, , \\
 \tr\left[ \bGhatnur \frac{\del \bGhat^{-1}(\nur)}{\del(\bmtomo)_{\ra \rb}} \right]& =
\sum_{\rA\rB\rc\rd}\bGhatnur_{\rA \rB} \left(-\nur^2 \left( \bgt  \basr \, \bc \right)_{\rB\rc} \frac{\del (\bmtomo+\nur^2 \bctomo)_{\rc\rd}}{\del (\bmtomo)_{\ra \rb}} \left( \bc \, \basr \, \bg \right)_{\rd\rA} \right) \, .
\end{align}
\end{widetext}
\newpage
Observe that:
\begin{align}
 \frac{\del (\bmtomo+\nur^2 \bctomo)_{\rc\rd}}{\del
(\bmtomo)_{\ra \rb}}=\delta_{\rc\ra}\delta_{\rd\rb} \, .
\end{align}
\\
\\
It follows that
\begin{align}
& \tr\left[ \bGhatnur \frac{\del \bGhat^{-1}(\nur)}{\del (\bmtomo)_{\ra \rb}} \right] \nonumber \\
&=- \sum_{\rA \rB} \nur^2\bGhatnur_{\rA \rB} \left( \bgt  \basr \bc \right)_{\rB\ra} \left(\bc \, \basr \, \bg\right)_{\rb\rA} \, ,\nonumber\\
&=- \left( \bc \, \basr \, \bg \, \nur^2\bGhatnur \, \bgt  \basr \, \bc \right)_{\rb\ra} \, .
\end{align}
In the last step we used the fact that both $\bGhatnur$ and $\basr$ are
symmetric matrices. Eq.~\eqref{eq:tr[GdelG]} becomes:
\begin{align}
\frac{\del}{\del (\bmtomo)} & \tr \log \! \left[ \bM^{-1} \bGhat^{-1}(\nur) \right] \nonumber \\
&=-\bc \, \basr \, \bg \, \nur^2 \bGhatnur \bgt \basr \, \bc \, .
\end{align}
We plug this into Eq.~\eqref{eq:ppf} and note that $\bc \, \basr = \blm \,
\bfhatnur$ to
get Eq.~\eqref{eq:<pp>}.

The derivation of Eq.~\eqref{eq:<xx>} is almost identical to that of
Eq.~\eqref{eq:<pp>} but with more terms. We do not show the details of that
derivation here.

\section{Proof of conclusions of Sec.~\ref{sec:generalequivalence}}
\label{app:multitime}

Using the fact that all position-momentum correlations vanish we get:
\begin{align}
 \lbr \bxi(s) \, \bXtz \rbr &= \bgt  \bfdot(s) \, \blm \lbr \bxz \, \bXtz \rbr \,
, \\
 \lbr \bxi(s) \, \bPtz \rbr &= \bgt  \blf(s) \lbr \bpz \, \bPtz \rbr \, ,
\end{align}
where the expectation values on the RHS are given by
Eqs.~(\ref{eq:<xX>},\ref{eq:<pP>}).

\begin{align}
\label{eq:last1}
 \lbr \bxi(s) \, \bXtz \rbr
&= \bgt \bmtomoot \frac{\cosomega s)}{\bomega^2 } \bmtomoot \bg \lbr \bXz \bXtz \rbr \\
&+ \frac{2}{\beta} \rsum \bgt \bmtomoot \frac{\cosomega s) \, \nur^2 }{\bomega^2 (\bomega^2 +\nur^2)} \bmtomoot \bg \, \bGhatnur \, . \nonumber
\end{align}
The first term on the right-hand side can be seen to decay by the fact that
\begin{equation}
\bgt \bmtomoot \frac{\cosomega s)}{\bomega^2 } \bmtomoot \bg = 2 \bgamma(s) \, .
\end{equation}
The second term can be seen to decay by noting the inequality
\begin{align}
\label{eq:kernelordering}
\nonumber
\bgt \bmtomoot \frac{\cosomega s) \, \nur^2}{\bomega^2 (\bomega^2+\nur^2)}\bmtomoot \bg
& \leq \bgt \bmtomoot \frac{\cosomega s)}{\bomega^2 }\bmtomoot \bg \, , \\
& \leq  2 \bgamma(s) \, ,
\end{align}
in the sense of positive-definite matrix kernels,
since both $\bomega^2$ and $(\bomega^2+\nur^2)$ are positive matrices and cosine is a positive-definite kernel.
The summation over $r$ in Eq.~\eqref{eq:last1} is finite as can be seen from
Eq.~\eqref{eq:<XX>}. As a result $\lbr \bxi(s) \, \bXtz \rbr$ is a function
that decays over time like $\bgamma(s)$. When we take the convolution of this
with another decaying function $\bGhat(t-s)$ and let $t\rightarrow \infty$ the
overlap goes to zero. This way we argue that second line of
Eq.~\eqref{eq:<X(t)X(t')>} vanishes. A similar calculation establishes the same
goes for the third line.
\begin{align}
&\lbr \bxi(s) \, \bPtz \rbr \nonumber \\
&= \bgt  \blf(s) \, \bc \lbr \bxz \bXtz \rbr \bM +\bgt \blf(s) \, \bg \lbr \bXz \bXtz \rbr \bM \, ,\\
&= \frac{1}{\beta} \rsum \bgt  \blf(s) \, \bc \, \bmtomoot \frac{\nur^2}{\bomega^2 (\bomega^2+\nur^2)} \bmtomoot \bg \, \bGhatnur \, \bM \, , \\
&= \frac{1}{\beta} \rsum \bgt \bmtomoot \frac{\sinomega s) \, \nur^2}{\bomega (\bomega^2+\nur^2)}\bmtomoot \bg \, \bGhatnur \, \bM \, , \\
&= - \frac{1}{\beta} \rsum \frac{\rd}{\rd s} \left[ \bgt \bmtomoot \frac{\cosomega s) \, \nur^2 }{\bomega^2 (\bomega^2+\nur^2)}\bmtomoot \bg \right] \, \bGhatnur \, \bM \, .
\end{align}
The term inside square brackets decays as $\bgamma(s)$ as can be seen
from Eq.~\eqref{eq:kernelordering} and the argument following it.
The summation over $r$ is finite as before. Hence $ \lbr \bxi(s) \, \bPtz \rbr$ decays over time like $\bgammadot(s)$.
The convolution of this with another decaying function $\bG(t-\tau)$ gives zero in the
limit $t\rightarrow \infty$.

The second and third lines of Eq.~\eqref{eq:<X(t)X(t')>} are zero for the
uncorrelated initial state as well. This follows trivially from: $\left \langle
\bxi(s) \,
\bXtz
\right \rangle_\rE=\left \langle \bxi(s) \, \bPtz
\right \rangle_\rE=0$.

Finally we need to show that the fourth line of Eq.~\eqref{eq:<X(t)X(t')>}
is the same for both cases. This requires showing that the late-time limit of
the noise kernel is the same. We know that the noise kernel is stationary for
the uncorrelated initial state. Let us focus on the noise kernel of the closed system
thermal state.
\begin{align}
\label{eq:<xixi>}
& \lbr \bxi(s) \, \bxi(s')^\rT \rbr = \\
&\bgt \left( \bfdot(s) \, \blm \lbr \bxz \bxtz \rbr \blm \, \bfdot(s') + \blf(s) \lbr \bpz \bptz \rbr \blf(s') \right) \bg \, . \nonumber
\end{align}
We use Eqs.~(\ref{eq:<pp>},\ref{eq:<xx>}) on the RHS.
The derivation is straightforward but tedious. The theorems
in App.~\ref{app:matrixtheorems} are
utilized repeatedly.

The uncorrelated noise kernel is obtained if only the first terms in
Eqs.~(\ref{eq:<pp>},\ref{eq:<xx>}) are kept and the rest ignored.
Hence we need to show that all the other terms vanish in the late-time limit.
The strategy is the same as before: we show that these terms are bounded by a
function proportional to the damping kernel or its derivatives. We work out
the details for two terms explicitly.

First consider the term in the noise kernel Eq.~\eqref{eq:<xixi>} due to the
second term in Eq.~\eqref{eq:<xx>}.
\begin{widetext}
\begin{align}
\bgt \bfdot(s) & \, \blm \, \bctomo \bg \lbr \bXz \bXtz \rbr \bgt \bctomo \blm \, \bfdot(s') \, \bg \nonumber \\
&= \bgt \bmtomoot \cosomega s) \, \bmtomoot \blm \, \bmtomoot
\bomega^{-2} \bmtomoot \bg \lbr \bXz \bXtz \rbr \bgt \bmtomoot \bomega^{-2}
\bmtomoot \blm \, \bmtomoot \cosomega s') \, \bmtomoot \bg \, ,\\
&= \bgt \bmtomoot \frac{\cosomega s)}{\bomega^2} \bmtomoot \bg \lbr \bXz \bXtz
\rbr \bgt \bmtomoot \frac{\cosomega s')}{\bomega^2} \bmtomoot \bg \, , \\
&= 4 \bgamma(s) \lbr \bXz \bXtz \rbr \bgamma(s') \, .
\end{align}
\end{widetext}
Unlike previous cases we were able to express this term exactly in terms of
the damping kernel. It is a decaying function in both $s$ and $s'$ variables.
The convolution of $\bgamma(s)$ with $\bGhat(t-s)$ in
Eq.~\eqref{eq:<X(t)X(t')>} goes to zero if we let $t\rightarrow \infty$.
Similarly the overlap of $\bgamma(s')$ with $\bGhat(t'-s')$ vanishes in the
limit $t' \rightarrow \infty$.

Secondly consider the term in the noise kernel Eq.~\eqref{eq:<xixi>} due to the
third term in Eq.~\eqref{eq:<xx>}.
\begin{widetext}
\begin{align}
\nonumber
-\frac{2}{\beta} \rsum & \bgt \bfdot(s) \, \blm \,  \bmtomoot
\frac{1}{\bomega^2(\bomega^2+\nur^2)} \bmtomoot
\bg \, \nur^2 \bGhatnur \,
\bgt \bmtomoot \frac{1}{\bomega^2(\bomega^2+\nur^2)} \bomega^2 \bmtomoot
\blm \, \bfdot(s') \, \bg \\
&= -\frac{2}{\beta} \rsum  \bgt \bmtomoot \frac{\cosomega s)}{\bomega^2
(\bomega^2 + \nur^2)} \bmtomoot
\bg \, \nur^2 \bGhatnur \,
\bgt \bmtomoot \frac{ \cosomega s')}{\bomega^2+\nur^2}
\bmtomoot  \bg \, , \\
&= \frac{2}{\beta} \rsum \left[ \bgt \bmtomoot \frac{\cosomega s) \,
\nur^2 }{\bomega^2
(\bomega^2 + \nur^2)} \bmtomoot
\bg \right] \, \frac{\bGhatnur}{\nur^2} \frac{\rd^2}{\rd s'^2} \, \left[
\bgt \bmtomoot \frac{ \cosomega s') \, \nur^2 }{\bomega^2(\bomega^2+\nur^2)}
\bmtomoot  \bg \right]\, .
\end{align}
\end{widetext}
As before we conclude that the terms in square brackets decay like the damping
kernel. The summation over $r$ is finite as can be seen from
Eq.~\eqref{eq:<XX>} and noting that $\nu_r > 1$ for all positive $r$.

Close inspection of all the other terms in Eq.~\eqref{eq:<xixi>} reveals
that
they have roughly the same form as those we worked out the details explicitly.
All
these terms vanish in the late-time limit.

This proves the equivalence of the late-time limit of the uncorrelated initial
state to that of the late-time limit of the closed system thermal state.
Since the closed system thermal state is stationary our proof is complete.

\section*{Acknowledgement}

One of the authors (YS) would like to thank M. E. Fisher and C. Jarzynski for useful discussions. YS, CHF and BLH were supported in part by NSF grants PHY-0801368 to the University of Maryland. JMT and CHF were supported in part by the NSF Physics Frontier Center at the JQI.
This work was finished while BLH was visiting the Institute for Advanced Study of the Hong Kong University of Science and Technology.

\bibliography{bib}{}

\begin{thebibliography}{38}%
\makeatletter
\providecommand \@ifxundefined [1]{%
 \@ifx{#1\undefined}
}%
\providecommand \@ifnum [1]{%
 \ifnum #1\expandafter \@firstoftwo
 \else \expandafter \@secondoftwo
 \fi
}%
\providecommand \@ifx [1]{%
 \ifx #1\expandafter \@firstoftwo
 \else \expandafter \@secondoftwo
 \fi
}%
\providecommand \natexlab [1]{#1}%
\providecommand \enquote  [1]{``#1''}%
\providecommand \bibnamefont  [1]{#1}%
\providecommand \bibfnamefont [1]{#1}%
\providecommand \citenamefont [1]{#1}%
\providecommand \href@noop [0]{\@secondoftwo}%
\providecommand \href [0]{\begingroup \@sanitize@url \@href}%
\providecommand \@href[1]{\@@startlink{#1}\@@href}%
\providecommand \@@href[1]{\endgroup#1\@@endlink}%
\providecommand \@sanitize@url [0]{\catcode `\\12\catcode `\$12\catcode
  `\&12\catcode `\#12\catcode `\^12\catcode `\_12\catcode `\%12\relax}%
\providecommand \@@startlink[1]{}%
\providecommand \@@endlink[0]{}%
\providecommand \url  [0]{\begingroup\@sanitize@url \@url }%
\providecommand \@url [1]{\endgroup\@href {#1}{\urlprefix }}%
\providecommand \urlprefix  [0]{URL }%
\providecommand \Eprint [0]{\href }%
\providecommand \doibase [0]{http://dx.doi.org/}%
\providecommand \selectlanguage [0]{\@gobble}%
\providecommand \bibinfo  [0]{\@secondoftwo}%
\providecommand \bibfield  [0]{\@secondoftwo}%
\providecommand \translation [1]{[#1]}%
\providecommand \BibitemOpen [0]{}%
\providecommand \bibitemStop [0]{}%
\providecommand \bibitemNoStop [0]{.\EOS\space}%
\providecommand \EOS [0]{\spacefactor3000\relax}%
\providecommand \BibitemShut  [1]{\csname bibitem#1\endcsname}%
\let\auto@bib@innerbib\@empty
\bibitem [{\citenamefont {Popescu}\ \emph {et~al.}(2006)\citenamefont
  {Popescu}, \citenamefont {Short},\ and\ \citenamefont {Winter}}]{Popescu06}%
  \BibitemOpen
  \bibfield  {author} {\bibinfo {author} {\bibfnamefont {S.}~\bibnamefont
  {Popescu}}, \bibinfo {author} {\bibfnamefont {A.~J.}\ \bibnamefont {Short}},
  \ and\ \bibinfo {author} {\bibfnamefont {A.}~\bibnamefont {Winter}},\ }\href
  {\doibase 10.1038/nphys444} {\bibfield  {journal} {\bibinfo  {journal} {Nat.
  Phys.}\ }\textbf {\bibinfo {volume} {2}},\ \bibinfo {pages} {754} (\bibinfo
  {year} {2006})}\BibitemShut {NoStop}%
\bibitem [{\citenamefont {Goldstein}\ \emph {et~al.}(2006)\citenamefont
  {Goldstein}, \citenamefont {Lebowitz}, \citenamefont {Tumulka},\ and\
  \citenamefont {Zangh\`\i}}]{Goldstein06}%
  \BibitemOpen
  \bibfield  {author} {\bibinfo {author} {\bibfnamefont {S.}~\bibnamefont
  {Goldstein}}, \bibinfo {author} {\bibfnamefont {J.~L.}\ \bibnamefont
  {Lebowitz}}, \bibinfo {author} {\bibfnamefont {R.}~\bibnamefont {Tumulka}}, \
  and\ \bibinfo {author} {\bibfnamefont {N.}~\bibnamefont {Zangh\`\i}},\ }\href
  {\doibase 10.1103/PhysRevLett.96.050403} {\bibfield  {journal} {\bibinfo
  {journal} {Phys. Rev. Lett.}\ }\textbf {\bibinfo {volume} {96}},\ \bibinfo
  {pages} {050403} (\bibinfo {year} {2006})}\BibitemShut {NoStop}%
\bibitem [{\citenamefont {van Hove}(1954)}]{vanHove55}%
  \BibitemOpen
  \bibfield  {author} {\bibinfo {author} {\bibfnamefont {L.}~\bibnamefont {van
  Hove}},\ }\href {\doibase 10.1016/S0031-8914(54)92646-4} {\bibfield
  {journal} {\bibinfo  {journal} {Physica}\ }\textbf {\bibinfo {volume} {21}},\
  \bibinfo {pages} {517 } (\bibinfo {year} {1954})}\BibitemShut {NoStop}%
\bibitem [{\citenamefont {Davies}(1974)}]{Davies74}%
  \BibitemOpen
  \bibfield  {author} {\bibinfo {author} {\bibfnamefont {E.~B.}\ \bibnamefont
  {Davies}},\ }\href {\doibase 10.1007/BF01608389} {\bibfield  {journal}
  {\bibinfo  {journal} {Comm. Math. Phys.}\ }\textbf {\bibinfo {volume} {39}},\
  \bibinfo {pages} {91} (\bibinfo {year} {1974})}\BibitemShut {NoStop}%
\bibitem [{\citenamefont {Davies}(1976{\natexlab{a}})}]{Davies76a}%
  \BibitemOpen
  \bibfield  {author} {\bibinfo {author} {\bibfnamefont {E.~B.}\ \bibnamefont
  {Davies}},\ }\href {\doibase 10.1007/BF01351898} {\bibfield  {journal}
  {\bibinfo  {journal} {Math. Ann.}\ }\textbf {\bibinfo {volume} {219}},\
  \bibinfo {pages} {147} (\bibinfo {year} {1976}{\natexlab{a}})}\BibitemShut
  {NoStop}%
\bibitem [{\citenamefont {Davies}(1976{\natexlab{b}})}]{Davies76}%
  \BibitemOpen
  \bibfield  {author} {\bibinfo {author} {\bibfnamefont {E.~B.}\ \bibnamefont
  {Davies}},\ }\href@noop {} {\emph {\bibinfo {title} {Quantum theory of open
  systems}}}\ (\bibinfo  {publisher} {Academic Press},\ \bibinfo {address}
  {London},\ \bibinfo {year} {1976})\BibitemShut {NoStop}%
\bibitem [{\citenamefont {Calzetta}\ and\ \citenamefont
  {Hu}(2008)}]{CalzettaHu08}%
  \BibitemOpen
  \bibfield  {author} {\bibinfo {author} {\bibfnamefont {E.~A.}\ \bibnamefont
  {Calzetta}}\ and\ \bibinfo {author} {\bibfnamefont {B.~L.}\ \bibnamefont
  {Hu}},\ }\href {\doibase 10.1017/CBO9780511535123} {\emph {\bibinfo {title}
  {Nonequilibrium Quantum Field Theory}}}\ (\bibinfo  {publisher} {Cambridge
  University Press},\ \bibinfo {address} {Cambridge},\ \bibinfo {year}
  {2008})\BibitemShut {NoStop}%
\bibitem [{\citenamefont {Balescu}(1997)}]{Balescu97}%
  \BibitemOpen
  \bibfield  {author} {\bibinfo {author} {\bibfnamefont {R.}~\bibnamefont
  {Balescu}},\ }\href@noop {} {\emph {\bibinfo {title} {Statistical Dynamics:
  Matter out of Equilibrium}}}\ (\bibinfo  {publisher} {Imperial College
  Press},\ \bibinfo {address} {London},\ \bibinfo {year} {1997})\BibitemShut
  {NoStop}%
\bibitem [{\citenamefont {Linden}\ \emph {et~al.}(2009)\citenamefont {Linden},
  \citenamefont {Popescu}, \citenamefont {Short},\ and\ \citenamefont
  {Winter}}]{Linden09}%
  \BibitemOpen
  \bibfield  {author} {\bibinfo {author} {\bibfnamefont {N.}~\bibnamefont
  {Linden}}, \bibinfo {author} {\bibfnamefont {S.}~\bibnamefont {Popescu}},
  \bibinfo {author} {\bibfnamefont {A.~J.}\ \bibnamefont {Short}}, \ and\
  \bibinfo {author} {\bibfnamefont {A.}~\bibnamefont {Winter}},\ }\href
  {\doibase 10.1103/PhysRevE.79.061103} {\bibfield  {journal} {\bibinfo
  {journal} {Phys. Rev. E}\ }\textbf {\bibinfo {volume} {79}},\ \bibinfo
  {pages} {061103} (\bibinfo {year} {2009})}\BibitemShut {NoStop}%
\bibitem [{\citenamefont {Reimann}(2010)}]{Reimann10}%
  \BibitemOpen
  \bibfield  {author} {\bibinfo {author} {\bibfnamefont {P.}~\bibnamefont
  {Reimann}},\ }\href {http://stacks.iop.org/1367-2630/12/i=5/a=055027}
  {\bibfield  {journal} {\bibinfo  {journal} {New Journal of Physics}\ }\textbf
  {\bibinfo {volume} {12}},\ \bibinfo {pages} {055027} (\bibinfo {year}
  {2010})}\BibitemShut {NoStop}%
\bibitem [{\citenamefont {Short}(2011)}]{Short11}%
  \BibitemOpen
  \bibfield  {author} {\bibinfo {author} {\bibfnamefont {A.~J.}\ \bibnamefont
  {Short}},\ }\href {http://stacks.iop.org/1367-2630/13/i=5/a=053009}
  {\bibfield  {journal} {\bibinfo  {journal} {New J. Phys.}\ }\textbf {\bibinfo
  {volume} {13}} (\bibinfo {year} {2011})}\BibitemShut {NoStop}%
\bibitem [{\citenamefont {Short}\ and\ \citenamefont
  {Farrelly}(2012)}]{Short12}%
  \BibitemOpen
  \bibfield  {author} {\bibinfo {author} {\bibfnamefont {A.~J.}\ \bibnamefont
  {Short}}\ and\ \bibinfo {author} {\bibfnamefont {T.~C.}\ \bibnamefont
  {Farrelly}},\ }\href {http://stacks.iop.org/1367-2630/14/i=1/a=013063}
  {\bibfield  {journal} {\bibinfo  {journal} {New Journal of Physics}\ }\textbf
  {\bibinfo {volume} {14}},\ \bibinfo {pages} {013063} (\bibinfo {year}
  {2012})}\BibitemShut {NoStop}%
\bibitem [{\citenamefont {Feynman}\ and\ \citenamefont
  {Vernon}(1963)}]{Feynman63}%
  \BibitemOpen
  \bibfield  {author} {\bibinfo {author} {\bibfnamefont {R.~P.}\ \bibnamefont
  {Feynman}}\ and\ \bibinfo {author} {\bibfnamefont {F.~L.}\ \bibnamefont
  {Vernon}},\ }\href {\doibase DOI:10.1016/0003-4916(63)90068-X} {\bibfield
  {journal} {\bibinfo  {journal} {Ann. Phys.}\ }\textbf {\bibinfo {volume}
  {24}},\ \bibinfo {pages} {118 } (\bibinfo {year} {1963})}\BibitemShut
  {NoStop}%
\bibitem [{\citenamefont {Caldeira}\ and\ \citenamefont
  {Leggett}(1983)}]{CaldeiraLeggett83}%
  \BibitemOpen
  \bibfield  {author} {\bibinfo {author} {\bibfnamefont {A.~O.}\ \bibnamefont
  {Caldeira}}\ and\ \bibinfo {author} {\bibfnamefont {A.~J.}\ \bibnamefont
  {Leggett}},\ }\href {\doibase DOI:10.1016/0378-4371(83)90013-4} {\bibfield
  {journal} {\bibinfo  {journal} {Physica A}\ }\textbf {\bibinfo {volume}
  {121}},\ \bibinfo {pages} {587 } (\bibinfo {year} {1983})}\BibitemShut
  {NoStop}%
\bibitem [{\citenamefont {Hu}\ \emph {et~al.}(1992)\citenamefont {Hu},
  \citenamefont {Paz},\ and\ \citenamefont {Zhang}}]{HPZ92}%
  \BibitemOpen
  \bibfield  {author} {\bibinfo {author} {\bibfnamefont {B.~L.}\ \bibnamefont
  {Hu}}, \bibinfo {author} {\bibfnamefont {J.~P.}\ \bibnamefont {Paz}}, \ and\
  \bibinfo {author} {\bibfnamefont {Y.}~\bibnamefont {Zhang}},\ }\href
  {\doibase 10.1103/PhysRevD.45.2843} {\bibfield  {journal} {\bibinfo
  {journal} {Phys. Rev. D}\ }\textbf {\bibinfo {volume} {45}},\ \bibinfo
  {pages} {2843} (\bibinfo {year} {1992})}\BibitemShut {NoStop}%
\bibitem [{\citenamefont {Weiss}(1993)}]{Weiss93}%
  \BibitemOpen
  \bibfield  {author} {\bibinfo {author} {\bibfnamefont {U.}~\bibnamefont
  {Weiss}},\ }\href@noop {} {\emph {\bibinfo {title} {Quantum Dissipative
  Systems}}}\ (\bibinfo  {publisher} {World Scientific},\ \bibinfo {address}
  {Singapore},\ \bibinfo {year} {1993})\BibitemShut {NoStop}%
\bibitem [{\citenamefont {Breuer}\ and\ \citenamefont
  {Petruccione}(2002)}]{Breuer02}%
  \BibitemOpen
  \bibfield  {author} {\bibinfo {author} {\bibfnamefont {H.~P.}\ \bibnamefont
  {Breuer}}\ and\ \bibinfo {author} {\bibfnamefont {F.}~\bibnamefont
  {Petruccione}},\ }\href {\doibase 10.1093/acprof:oso/9780199213900.001.0001}
  {\emph {\bibinfo {title} {The Theory of Open Quantum Systems}}}\ (\bibinfo
  {publisher} {Oxford University Press},\ \bibinfo {address} {New York},\
  \bibinfo {year} {2002})\BibitemShut {NoStop}%
\bibitem [{\citenamefont {Fleming}\ \emph
  {et~al.}(2010{\natexlab{a}})\citenamefont {Fleming}, \citenamefont {Hu},\
  and\ \citenamefont {Roura}}]{FDR}%
  \BibitemOpen
  \bibfield  {author} {\bibinfo {author} {\bibfnamefont {C.~H.}\ \bibnamefont
  {Fleming}}, \bibinfo {author} {\bibfnamefont {B.~L.}\ \bibnamefont {Hu}}, \
  and\ \bibinfo {author} {\bibfnamefont {A.}~\bibnamefont {Roura}},\
  }\href@noop {} {\enquote {\bibinfo {title} {Non-equilibrium
  fluctuation-dissipation inequality, and non-equilibrium uncertainty
  principle},}\ } (\bibinfo {year} {2010}{\natexlab{a}}),\ \Eprint
  {http://arxiv.org/abs/1012.0681} {arXiv:1012.0681 [quant-ph]} \BibitemShut
  {NoStop}%
\bibitem [{\citenamefont {Geva}\ \emph {et~al.}(2000)\citenamefont {Geva},
  \citenamefont {Rosenman},\ and\ \citenamefont {Tannor}}]{Geva00}%
  \BibitemOpen
  \bibfield  {author} {\bibinfo {author} {\bibfnamefont {E.}~\bibnamefont
  {Geva}}, \bibinfo {author} {\bibfnamefont {E.}~\bibnamefont {Rosenman}}, \
  and\ \bibinfo {author} {\bibfnamefont {D.}~\bibnamefont {Tannor}},\ }\href
  {\doibase 10.1063/1.481928} {\bibfield  {journal} {\bibinfo  {journal} {J.
  Chem. Phys.}\ }\textbf {\bibinfo {volume} {113}},\ \bibinfo {pages} {1380}
  (\bibinfo {year} {2000})}\BibitemShut {NoStop}%
\bibitem [{\citenamefont {Tasaki}\ \emph {et~al.}(2007)\citenamefont {Tasaki},
  \citenamefont {Yuasa}, \citenamefont {Facchi}, \citenamefont {Kimura},
  \citenamefont {Nakazato}, \citenamefont {Ohba},\ and\ \citenamefont
  {Pascazio}}]{Tasaki07}%
  \BibitemOpen
  \bibfield  {author} {\bibinfo {author} {\bibfnamefont {S.}~\bibnamefont
  {Tasaki}}, \bibinfo {author} {\bibfnamefont {K.}~\bibnamefont {Yuasa}},
  \bibinfo {author} {\bibfnamefont {P.}~\bibnamefont {Facchi}}, \bibinfo
  {author} {\bibfnamefont {G.}~\bibnamefont {Kimura}}, \bibinfo {author}
  {\bibfnamefont {H.}~\bibnamefont {Nakazato}}, \bibinfo {author}
  {\bibfnamefont {I.}~\bibnamefont {Ohba}}, \ and\ \bibinfo {author}
  {\bibfnamefont {S.}~\bibnamefont {Pascazio}},\ }\href {\doibase
  10.1016/j.aop.2006.06.004} {\bibfield  {journal} {\bibinfo  {journal} {Ann.
  Phys.}\ }\textbf {\bibinfo {volume} {322}},\ \bibinfo {pages} {631} (\bibinfo
  {year} {2007})}\BibitemShut {NoStop}%
\bibitem [{\citenamefont {Fleming}\ \emph
  {et~al.}(2011{\natexlab{a}})\citenamefont {Fleming}, \citenamefont {Roura},\
  and\ \citenamefont {Hu}}]{Correlations}%
  \BibitemOpen
  \bibfield  {author} {\bibinfo {author} {\bibfnamefont {C.~H.}\ \bibnamefont
  {Fleming}}, \bibinfo {author} {\bibfnamefont {A.}~\bibnamefont {Roura}}, \
  and\ \bibinfo {author} {\bibfnamefont {B.~L.}\ \bibnamefont {Hu}},\ }\href
  {\doibase 10.1103/PhysRevE.84.021106} {\bibfield  {journal} {\bibinfo
  {journal} {Phys. Rev. E}\ }\textbf {\bibinfo {volume} {84}},\ \bibinfo
  {pages} {021106} (\bibinfo {year} {2011}{\natexlab{a}})}\BibitemShut
  {NoStop}%
\bibitem [{\citenamefont {Fleming}\ and\ \citenamefont {Hu}()}]{QOS}%
  \BibitemOpen
  \bibfield  {author} {\bibinfo {author} {\bibfnamefont {C.~H.}\ \bibnamefont
  {Fleming}}\ and\ \bibinfo {author} {\bibfnamefont {B.~L.}\ \bibnamefont
  {Hu}},\ }\href@noop {} {\enquote {\bibinfo {title} {Non-{M}arkovian dynamics
  of open quantum systems: Stochastic equations and their perturbative
  solutions},}\ }\BibitemShut {NoStop}%
\bibitem [{\citenamefont {Mori}\ and\ \citenamefont
  {Miyashita}(2008)}]{Mori08}%
  \BibitemOpen
  \bibfield  {author} {\bibinfo {author} {\bibfnamefont {T.}~\bibnamefont
  {Mori}}\ and\ \bibinfo {author} {\bibfnamefont {S.}~\bibnamefont
  {Miyashita}},\ }\href {\doibase 10.1143/JPSJ.77.124005} {\bibfield  {journal}
  {\bibinfo  {journal} {J. Phys. Soc. Jap.}\ }\textbf {\bibinfo {volume}
  {77}},\ \bibinfo {pages} {124005} (\bibinfo {year} {2008})}\BibitemShut
  {NoStop}%
\bibitem [{\citenamefont {Swain}(1981)}]{Swain81}%
  \BibitemOpen
  \bibfield  {author} {\bibinfo {author} {\bibfnamefont {S.}~\bibnamefont
  {Swain}},\ }\href {http://stacks.iop.org/0305-4470/14/i=10/a=013} {\bibfield
  {journal} {\bibinfo  {journal} {J. Phys. A}\ }\textbf {\bibinfo {volume}
  {14}},\ \bibinfo {pages} {2577} (\bibinfo {year} {1981})}\BibitemShut
  {NoStop}%
\bibitem [{\citenamefont {Fleming}\ \emph
  {et~al.}(2011{\natexlab{b}})\citenamefont {Fleming}, \citenamefont {Roura},\
  and\ \citenamefont {Hu}}]{NQBM}%
  \BibitemOpen
  \bibfield  {author} {\bibinfo {author} {\bibfnamefont {C.~H.}\ \bibnamefont
  {Fleming}}, \bibinfo {author} {\bibfnamefont {A.}~\bibnamefont {Roura}}, \
  and\ \bibinfo {author} {\bibfnamefont {B.~L.}\ \bibnamefont {Hu}},\
  }\href@noop {} {\enquote {\bibinfo {title} {Quantum brownian motion of
  multipartite systems with entanglement dynamics},}\ } (\bibinfo {year}
  {2011}{\natexlab{b}}),\ \Eprint {http://arxiv.org/abs/1106.5752}
  {arXiv:1106.5752 [quant-ph]} \BibitemShut {NoStop}%
\bibitem [{\citenamefont {Fleming}()}]{Fleming11T}%
  \BibitemOpen
  \bibfield  {author} {\bibinfo {author} {\bibfnamefont {C.~H.}\ \bibnamefont
  {Fleming}},\ }\href {http://hdl.handle.net/1903/11708} {\enquote {\bibinfo
  {title} {Non-{M}arkovian dynamics of open quantum systems},}\ }\bibinfo
  {note} {{P}h.{D}. thesis, University of Maryland College Park
  (2011)}\BibitemShut {NoStop}%
\bibitem [{\citenamefont {Ford}\ \emph {et~al.}(1988)\citenamefont {Ford},
  \citenamefont {Lewis},\ and\ \citenamefont {O'Connell}}]{FordOconnell88}%
  \BibitemOpen
  \bibfield  {author} {\bibinfo {author} {\bibfnamefont {G.~W.}\ \bibnamefont
  {Ford}}, \bibinfo {author} {\bibfnamefont {J.~T.}\ \bibnamefont {Lewis}}, \
  and\ \bibinfo {author} {\bibfnamefont {R.~F.}\ \bibnamefont {O'Connell}},\
  }\href {\doibase 10.1103/PhysRevA.37.4419} {\bibfield  {journal} {\bibinfo
  {journal} {Phys. Rev. A}\ }\textbf {\bibinfo {volume} {37}},\ \bibinfo
  {pages} {4419} (\bibinfo {year} {1988})}\BibitemShut {NoStop}%
\bibitem [{\citenamefont {Calzetta}\ \emph {et~al.}(2003)\citenamefont
  {Calzetta}, \citenamefont {Roura},\ and\ \citenamefont {Verdaguer}}]{CRV03}%
  \BibitemOpen
  \bibfield  {author} {\bibinfo {author} {\bibfnamefont {E.}~\bibnamefont
  {Calzetta}}, \bibinfo {author} {\bibfnamefont {A.}~\bibnamefont {Roura}}, \
  and\ \bibinfo {author} {\bibfnamefont {E.}~\bibnamefont {Verdaguer}},\ }\href
  {\doibase DOI:10.1016/S0378-4371(02)01521-2} {\bibfield  {journal} {\bibinfo
  {journal} {Physica A}\ }\textbf {\bibinfo {volume} {319}},\ \bibinfo {pages}
  {188 } (\bibinfo {year} {2003})}\BibitemShut {NoStop}%
\bibitem [{\citenamefont {Kampen}\ and\ \citenamefont
  {Oppenheim}(1997)}]{Kampen97}%
  \BibitemOpen
  \bibfield  {author} {\bibinfo {author} {\bibfnamefont {N.}~\bibnamefont
  {Kampen}}\ and\ \bibinfo {author} {\bibfnamefont {I.}~\bibnamefont
  {Oppenheim}},\ }\href {\doibase 10.1007/BF02181287} {\bibfield  {journal}
  {\bibinfo  {journal} {J. Stat. Phys.}\ }\textbf {\bibinfo {volume} {87}},\
  \bibinfo {pages} {1325} (\bibinfo {year} {1997})}\BibitemShut {NoStop}%
\bibitem [{\citenamefont {Breuer}\ \emph {et~al.}(2003)\citenamefont {Breuer},
  \citenamefont {Ma},\ and\ \citenamefont {Petruccione}}]{Breuer03}%
  \BibitemOpen
  \bibfield  {author} {\bibinfo {author} {\bibfnamefont {H.~P.}\ \bibnamefont
  {Breuer}}, \bibinfo {author} {\bibfnamefont {A.}~\bibnamefont {Ma}}, \ and\
  \bibinfo {author} {\bibfnamefont {F.}~\bibnamefont {Petruccione}},\ }in\
  \href@noop {} {\emph {\bibinfo {booktitle} {Quantum Computing and Quantum
  Bits in Mesoscopic Systems}}},\ \bibinfo {editor} {edited by\ \bibinfo
  {editor} {\bibfnamefont {A.~J.}\ \bibnamefont {Leggett}}, \bibinfo {editor}
  {\bibfnamefont {B.}~\bibnamefont {Ruggiero}}, \ and\ \bibinfo {editor}
  {\bibfnamefont {P.}~\bibnamefont {Silvestrini}}}\ (\bibinfo  {publisher}
  {Kluwer},\ \bibinfo {address} {Dordrecht},\ \bibinfo {year} {2003})\ \Eprint
  {http://arxiv.org/abs/quant-ph/0209153} {quant-ph/0209153} \BibitemShut
  {NoStop}%
\bibitem [{\citenamefont {Strunz}\ and\ \citenamefont {Yu}(2004)}]{Strunz04}%
  \BibitemOpen
  \bibfield  {author} {\bibinfo {author} {\bibfnamefont {W.~T.}\ \bibnamefont
  {Strunz}}\ and\ \bibinfo {author} {\bibfnamefont {T.}~\bibnamefont {Yu}},\
  }\href {\doibase 10.1103/PhysRevA.69.052115} {\bibfield  {journal} {\bibinfo
  {journal} {Phys. Rev. A}\ }\textbf {\bibinfo {volume} {69}},\ \bibinfo
  {pages} {052115} (\bibinfo {year} {2004})}\BibitemShut {NoStop}%
\bibitem [{\citenamefont {Fleming}\ and\ \citenamefont
  {Cummings}(2011)}]{Accuracy}%
  \BibitemOpen
  \bibfield  {author} {\bibinfo {author} {\bibfnamefont {C.~H.}\ \bibnamefont
  {Fleming}}\ and\ \bibinfo {author} {\bibfnamefont {N.~I.}\ \bibnamefont
  {Cummings}},\ }\href {\doibase 10.1103/PhysRevE.83.031117} {\bibfield
  {journal} {\bibinfo  {journal} {Phys. Rev. E}\ }\textbf {\bibinfo {volume}
  {83}},\ \bibinfo {pages} {031117} (\bibinfo {year} {2011})}\BibitemShut
  {NoStop}%
\bibitem [{\citenamefont {Fleming}\ \emph
  {et~al.}(2010{\natexlab{b}})\citenamefont {Fleming}, \citenamefont
  {Cummings}, \citenamefont {Anastopoulos},\ and\ \citenamefont {Hu}}]{Dipole}%
  \BibitemOpen
  \bibfield  {author} {\bibinfo {author} {\bibfnamefont {C.~H.}\ \bibnamefont
  {Fleming}}, \bibinfo {author} {\bibfnamefont {N.~I.}\ \bibnamefont
  {Cummings}}, \bibinfo {author} {\bibfnamefont {C.}~\bibnamefont
  {Anastopoulos}}, \ and\ \bibinfo {author} {\bibfnamefont {B.~L.}\
  \bibnamefont {Hu}},\ }\href@noop {} {\enquote {\bibinfo {title}
  {Non-{M}arkovian dynamics and entanglement of two-level atoms in a common
  field},}\ } (\bibinfo {year} {2010}{\natexlab{b}}),\ \bibinfo {note}
  {accepted for publication: J. Phys. A},\ \Eprint
  {http://arxiv.org/abs/1012.5067} {arXiv:1012.5067 [quant-ph]} \BibitemShut
  {NoStop}%
\bibitem [{\citenamefont {Hilt}\ \emph {et~al.}(2011)\citenamefont {Hilt},
  \citenamefont {Thomas},\ and\ \citenamefont {Lutz}}]{Hilt2011}%
  \BibitemOpen
  \bibfield  {author} {\bibinfo {author} {\bibfnamefont {S.}~\bibnamefont
  {Hilt}}, \bibinfo {author} {\bibfnamefont {B.}~\bibnamefont {Thomas}}, \ and\
  \bibinfo {author} {\bibfnamefont {E.}~\bibnamefont {Lutz}},\ }\href {\doibase
  10.1103/PhysRevE.84.031110} {\bibfield  {journal} {\bibinfo  {journal} {Phys.
  Rev. E}\ }\textbf {\bibinfo {volume} {84}},\ \bibinfo {pages} {031110}
  (\bibinfo {year} {2011})}\BibitemShut {NoStop}%
\bibitem [{\citenamefont {Caldeira}\ and\ \citenamefont
  {Leggett}(1981)}]{CaldeiraLeggett81}%
  \BibitemOpen
  \bibfield  {author} {\bibinfo {author} {\bibfnamefont {A.~O.}\ \bibnamefont
  {Caldeira}}\ and\ \bibinfo {author} {\bibfnamefont {A.~J.}\ \bibnamefont
  {Leggett}},\ }\href {\doibase 10.1103/PhysRevLett.46.211} {\bibfield
  {journal} {\bibinfo  {journal} {Phys. Rev. Lett.}\ }\textbf {\bibinfo
  {volume} {46}},\ \bibinfo {pages} {211} (\bibinfo {year} {1981})}\BibitemShut
  {NoStop}%
\bibitem [{\citenamefont {Hillery}\ \emph {et~al.}(1984)\citenamefont
  {Hillery}, \citenamefont {O'Connell}, \citenamefont {Scully},\ and\
  \citenamefont {Wigner}}]{Hillery84}%
  \BibitemOpen
  \bibfield  {author} {\bibinfo {author} {\bibfnamefont {M.}~\bibnamefont
  {Hillery}}, \bibinfo {author} {\bibfnamefont {R.~F.}\ \bibnamefont
  {O'Connell}}, \bibinfo {author} {\bibfnamefont {M.~O.}\ \bibnamefont
  {Scully}}, \ and\ \bibinfo {author} {\bibfnamefont {E.~P.}\ \bibnamefont
  {Wigner}},\ }\href {\doibase DOI:10.1016/0370-1573(84)90160-1} {\bibfield
  {journal} {\bibinfo  {journal} {Phys. Rep.}\ }\textbf {\bibinfo {volume}
  {106}},\ \bibinfo {pages} {121 } (\bibinfo {year} {1984})}\BibitemShut
  {NoStop}%
\bibitem [{\citenamefont {Kirkwood}(1935)}]{Kirkwood35}%
  \BibitemOpen
  \bibfield  {author} {\bibinfo {author} {\bibfnamefont {J.~G.}\ \bibnamefont
  {Kirkwood}},\ }\href@noop {} {\bibfield  {journal} {\bibinfo  {journal} {J.
  Chem. Phys.}\ }\textbf {\bibinfo {volume} {3}},\ \bibinfo {pages} {300}
  (\bibinfo {year} {1935})}\BibitemShut {NoStop}%
\bibitem [{\citenamefont {Dwyer}(1967)}]{Dwyer67}%
  \BibitemOpen
  \bibfield  {author} {\bibinfo {author} {\bibfnamefont {P.~S.}\ \bibnamefont
  {Dwyer}},\ }\href@noop {} {\bibfield  {journal} {\bibinfo  {journal} {Journal
  of the American Statistical Assosiation}\ }\textbf {\bibinfo {volume} {62}},\
  \bibinfo {pages} {607} (\bibinfo {year} {1967})}\BibitemShut {NoStop}%
\end{thebibliography}%
\bibliographystyle{apsrev4-1}

\end{document}